\documentclass[preprint,12pt]{elsarticle}

\makeatletter
\def\ps@pprintTitle{%
 \let\@oddhead\@empty
 \let\@evenhead\@empty
 \def\@oddfoot{}%
 \let\@evenfoot\@oddfoot}
\makeatother
 
\newcommand{\abs}[1]{\left\lvert#1\right\rvert}
\newcommand{\norm}[1]{\left\lVert#1\right\rVert}

\usepackage[table]{xcolor}

\usepackage{graphicx}
\usepackage{epsfig}
\usepackage{subfigure}
\usepackage{wrapfig}
\usepackage{bm}% bold math
\usepackage{dcolumn}% Align table columns on decimal point
\usepackage{amsmath, amsthm, amssymb}

\usepackage{multirow}

\usepackage[]{hyperref}

\newtheorem{theorem}{Theorem}[section]

\newenvironment{definition}[1][Definition]{\begin{trivlist}
\item[\hskip \labelsep {\bfseries #1}]}{\end{trivlist}}

\journal{Computational Physics}

\begin{document}

\begin{frontmatter}

\title{A 3D fast algorithm for computing Lagrangian coherent structures via ridge tracking}

 \author[mae,CU]{Doug Lipinski}
 \author[mae,ece,CU]{Kamran Mohseni}

 \address[mae]{Dept. of Mechanical and Aerospace Engineering, University of Florida, P.O. Box 116250, 231 MAE-A, Gainesville, FL 32611}
 \address[ece]{Dept. of Electrical and Computer Engineering, University of Florida, P.O. Box 116200 
216 Larsen Hall 
Gainesville, FL 32611}
\address[CU]{This work was begun while the authors were affiliated with the Departments of Applied Mathematics and Aerospace Engineering Sciences at the University of Colorado - Boulder.}
 
% \thanks{Ph.D. Candidate, Applied Mathematics, University of Colorado - Boulder; Research Assistant, Dept. %of Mechanical and Aerospace Engineering, University of Florida - Gainesville}
%          and Kamran Mohseni\thanks{Prof., Dept. of Mechanical and Aerospace Engineering and Dept. of Electrical and Computer Engineering, University of Florida - Gainesville}}

\begin{abstract}
Lagrangian coherent structures (LCS) in fluid flows appear as co-dimension one ridges of the finite time Lyapunov exponent (FTLE) field. In three-dimensions this means two-dimensional ridges. A fast algorithm is presented here to locate and extract such ridge surfaces while avoiding unnecessary computations away from the LCS. This algorithm reduces the order of the computational complexity from $\mathcal{O}(1/dx^3)$ to about $\mathcal{O}(1/dx^{2})$ by eliminating computations over most of the three dimensional domain and computing the FTLE only near the two-dimensional ridge surfaces. The algorithm is grid based and proofs of error bounds for ridge locations are included. The algorithm performance and error bounds are verified in several examples. The algorithm offers significant advantages in computational cost as well as later data analysis.
\end{abstract}


\end{frontmatter}

\section{Introduction}

Lagrangian coherent structures (LCS) have seen increasingly popular use for visualizing and quantifying fluid structures and transport and mixing behavior in fluid flows. LCS were proposed by Haller and Yuan as the locally most repelling or attracting material lines in a flow~\cite{HallerG:00a}. Shadden et al.~\cite{Marsden:05g} later defined LCS as ridges of the finite time Lyapunov exponent (FTLE) field. Throughout this paper we will use this definition of LCS (ridges of the FTLE field), but other definitions of LCS will be discussed in the conclusions. LCS are known to have many useful properties such as denoting barriers to transport and repelling or attracting material surfaces~\cite{Marsden:05g}. Additionally, they form unambiguous boundaries to well known coherent structures such as vortices~\cite{Marsden:06a,Mohseni:08f} and are relatively insensitive to small errors~\cite{HallerG:02c}.

The FTLE field is typically computed numerically by integrating particle trajectories through the flow to approximate the flow map and then using finite differences to approximate the Jacobian of the flow map. The flow map which maps particles from their initial position at time $t_0$ to a final position at time $t_0+T$, $\Phi$, is given by 
\begin{equation}
\Phi_{t_0}^{t_0+T}(\mathbf{x}) = \mathbf{x}(t_0)+\int_{t_0}^{t_0+T}\mathbf{v}(\mathbf{x}(t)) dt
\end{equation}
and may be computed from analytical, experimental, or numerical velocity data. If the velocity data is not analytically defined it is usually necessary to read the velocity from data files and perform interpolations in space and time. The Jacobian of $\Phi$ is used to compute the deformation tensor, $\Delta$, which contains information about the stretching in the flow,
\begin{equation}
\Delta = \left(\frac{d\Phi}{d\mathbf{x}}\right)^{*}\left(\frac{d\Phi}{d\mathbf{x}}\right).
\end{equation}
Finally, the largest eigenvalue of $\Delta$ is used to define the finite time Lyapunov exponent, $\sigma$,
\begin{equation}
\sigma_{t_0}^T(\mathbf{x}) = \frac{1}{|T|}\ln\sqrt{\lambda_{max}(\Delta)}.
\end{equation}
Large values of the FTLE correspond to large amounts of stretching in the flow. Additionally, one may compute the flow map either forward or backward in time (positive or negative $T$). Ridges in the FTLE field are then expected to correspond to either the locally most attracting or repelling lines in the flow. This is not strictly true since it is also possible for high shear regions to exhibit large FTLE values, but experience has shown that FTLE ridges are often sufficient to learn about the underlying flow structure. More advanced techniques may be used to ensure that shear structures are not selected or that the resulting LCS are exact barriers to transport (e.g. Haller,~\cite{HallerG:11a}).

One of the largest hurdles to more widespread use of LCS techniques is the large time required to compute the FTLE. For typical fluid flows, computing the FTLE field requires advecting large numbers of particles through the flow. This must be done for each time step at which the FTLE field is desired. These particle advections dominate the total computational cost of any LCS algorithm.

More efficient algorithms are needed to address the large cost associated with LCS computations. Several past papers have attempted to address this problem with adaptive mesh refinement (AMR) algorithms that refine the computational mesh near the LCS~\cite{GarthC:07a, SadloF:07a}. This is effective, but still results in computing many FTLE values away from the LCS.

Additional attempts have been made to reuse computations to compute the LCS at subsequent time steps~\cite{RowleyCW:10}. Since the LCS are often desired at many different times and the time step may be smaller than the integration time used this leads to overlapping integrations. For example, if the LCS are desired at times $t=\{0.0,0.1,0.2,...10.0\}$ and an integration time of $T=1.0$ is to be used, the standard approach would be to compute a series of flow maps $\{\Phi_{0.0}^{10.0},\Phi_{0.1}^{10.1},\Phi_{0.2}^{10.2},\ldots\}$ for each time the FTLE field is required. However, it is instead possible to compute the flow maps $\{\Phi_{0.0}^{0.1},\Phi_{0.1}^{0.2},\Phi_{0.2}^{0.3},\ldots\}$ and use the fact that $\Phi_{0.0}^{10.0}=\Phi_{9.9}^{10.0}\circ\Phi_{9.8}^{9.9}\circ\Phi_{9.7}^{9.8}\cdots\Phi_{0.0}^{0.1}$. This composition of flow maps technique has the potential for large efficiency gains, but only if there is a significant overlap in the integration times and many time steps are desired. The method also comes at the cost of greatly increased memory usage to store all the necessary flow maps.

Another recent paper has reformulated the FTLE problem as an Eulerian level set problem involving the solution of a Liouville equation~\cite{LeungS:11}. This allows the use of any previously developed high order accurate schemes for the resulting hyperbolic PDE's. Although there are many existing techniques for solving such hyperbolic systems, the time step used in the Eulerian method must obey a CFL condition to ensure stability while the time step in more commonly used Lagrangian techniques is typically determined by the required accuracy. This typically means that the Lagrangian techniques are faster.

We propose another alternative to any of these approaches: detect and track the ridges on the fly. A recent publication discusses a gridless ridge tracking algorithm for computing the FTLE ridges in 2D flows~\cite{Mohseni:10g}. However, despite it's effectiveness in 2D (providing speedups of up to $80\times$), that  technique is not accompanied by proofs of convergence or error bounds on the results and does not readily generalize to higher dimensions. This is because one-dimensional ridges in a 2D flow may be represented as a simple curve and need only be tracked in two directions, but in three or more dimensions, the ridges are at least two-dimensional surfaces which require a more sophisticated representation.

In this paper, we present a fast algorithm for efficiently computing LCS in three-dimensional flows. In fact, most of the algorithm may be used for $n>0$ dimensions, but the surface triangulation used here is specific to 2D surfaces in a 3D space. All computations are performed on a predetermined orthogonal grid to simplify implementation and surface triangulation. A grid-less algorithm requires significant amounts of time to generate surface meshes, but by using a fixed grid, we are able to efficiently generate a surface triangulation via a lookup table similar to the marching cubes algorithm that is used for computing isosurfaces~\cite{LorensenWE:87a}. FTLE ridges are initially detected by computing the FTLE values on a series of lines across the domain. Local maxima along these lines occur at the FTLE ridge crossings. Once a series of points on the ridges have been detected, nearby points are tested to see if they are also on the ridge. This process is repeated to track the ridges through the entire domain. By performing computations only near the LCS surfaces, the order of the algorithm is reduced from $\mathcal{O}(1/dx^3)$ to about $\mathcal{O}(1/dx^2)$.

We present the results from several test cases, including a time dependent double gyre, Arnold-Beltrami-Childress (ABC) flow, and a swimming jellyfish. These results establish the computational order of the algorithm. We also investigate the computation cost for an FTLE field where the surface area of the LCS is known a priori and find that for a fixed grid spacing, the computational time is $C+\mathcal{O}(A_{LCS})$ where $C$ is a constant initialization cost and $A_{LCS}$ is the area of the LCS surfaces. Finally, the ridge tracking algorithm offers several other advantages beyond the savings in computational time. Since the LCS surfaces are directly computed, visualization of the LCS is simplified

\section{A surface tracking algorithm}
\label{sec:algorithm}

The computational savings seen by using a ridge tracking algorithm come from avoiding unnecessary computations away from the FTLE ridges. We separate this process into three steps: detecting initial points on the ridge surfaces, tracking the ridges through space, and triangulating the ridge points into a ridge surface. The initial ridge detection is handled by detecting where lines through the domain cross the ridge surfaces. The ridges are then iteratively grown by searching for nearby ridge points until no new ridge points are found. Finally, the use of a gridded coordinate system throughout this process allows the use of a lookup table to efficiently generate a triangulation of the resulting LCS surfaces.

To begin, we need a ridge definition. Shadden et al.~\cite{Marsden:05g} offers two definitions of ridges in two dimensions. Here, we extend the concept of second derivative ridges to $n$ dimensions:
\begin{definition}
{\it Ridge}: A {\it ridge} of a $C^2$ function $F$ is a co-dimension one surface $S$ satisfying
\begin{enumerate}
\item The vectors $\mathbf{n} \cdot \nabla F=0$ for all points on $S$ where $\mathbf{n}$ is a unit vector normal to $S$.
\item $\mathbf{n}^TH\mathbf{n}=\min_{\norm{\mathbf{u}}=1}\left(\mathbf{u}^TH\mathbf{u}\right)<0$ for all points on $S$ where $H$ is the Hessian matrix associated with $F$.
\end{enumerate}
\end{definition}

\subsection{Initial ridge detection}
\label{sec:initial_detection}

\begin{figure}
\centering
 \includegraphics[width=.4 \textwidth]{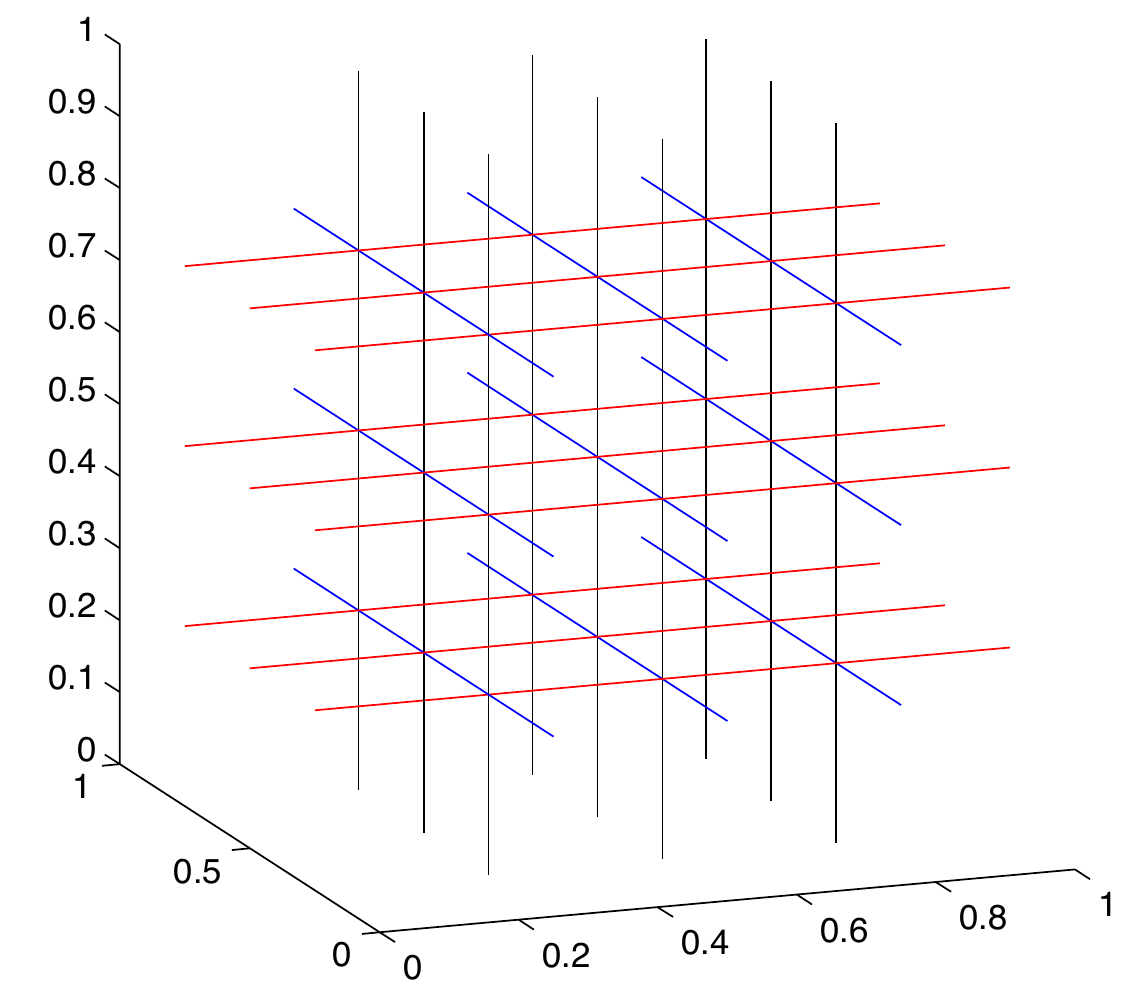}
 \caption{Initial ridge detection is handled by looking for local FTLE maxima along a few lines through the domain.}
 \label{fig:grid}
\end{figure}

If a hiker walks in a straight line, constantly monitoring his altitude, he will reach a locally maximum altitude upon crossing a ridge in the terrain. Similarly, if the FTLE is known along a line through a three-dimensional domain, local maxima along the line occur where the line crosses FTLE ridges. We restrict our computations to a fixed grid with spacing $dx$ and initially detect the FTLE ridges by computing the FTLE values along a set of lines that cross the domain as seen in Fig. \ref{fig:grid}. The number and spacing of the lines is determined by the user and depends on the expected spatial extend of the LCS in the flow. If an LCS surface does not intersect any of the lines used and is isolated from other LCS it may be missed entirely.

Once the FTLE values along these lines are known, we look for local maxima on each line and define grid ridge points as follows:
\begin{definition}
{\it Grid Ridge Point}:
Given an orthogonal grid in $\mathbb{R}^n$ with coordinate directions $\mathbf{e}_i: i\in\{1,...,n\}$; a grid point $\mathbf{x}_0$ is a {\it grid ridge point} of a function $F$ if $F(\mathbf{x}_0) \ge F(\mathbf{x}_0 \pm \mathbf{e}_i)$ for at least one $i\in\{1,...,n\}$.
\end{definition}
The local maxima on each line are grid ridge points and we begin tracking the FTLE ridges from these points. In the limit as grid spacing goes to zero, the grid ridge points converge to coordinate local maxima defined as:
\begin{definition}
{\it Coordinate local maximum}:
A point $\mathbf{x}_0$ is a coordinate local maximum of the function $F$ with respect to coordinate direction $\mathbf{e}_i$ if there is a value $\varepsilon>0$ such that $F(\mathbf{x}_0)>F(\mathbf{x}_0+\delta \cdot \mathbf{e}_i)$ for all $\delta<\varepsilon$.
\end{definition}
This property is proven below in Section \ref{sec:properties}.

It is also desirable to set a threshold for the FTLE ridge values at this time. Only detecting ridges with FTLE values above some threshold ensures that only the strongest LCS are revealed. This is often done by restricting the ridges to have FTLE values above a percentage of the maximum FTLE values that are detected. Typically $50-70\%$ of the maximum value is an adequate choice. Although the grid ridge points may exhibit false positives in the sense that such points may not necessarily correspond to FTLE ridges as defined by Shadden et al.~\cite{Marsden:05g}, experience has shown that this happens infrequently and typically does not alter the topology of the resulting ridge surfaces.

\subsection{Ridge tracking}

The heart of this algorithm lies in tracking the ridges outward from the initially detected grid ridge points. A schematic of this process for a two dimensional example is shown in Fig. \ref{fig:process}. The neighbors of each initially detected grid ridge point are checked to see if any meet the criteria to be grid ridge points (Fig. \ref{fig:process}b). In 3D, points in the grid are indexed by $(i,j,k)$ for the $x$, $y$, and $z$ directions. For each newly detected grid ridge point, $(i,j,k)$, the neighboring points in $[i-1,i+1] \times [j-1,j+1] \times [k-1,k+1]$ are checked to see if any are grid ridge points. If new grid ridge points are found that lie above the FTLE threshold the neighbors of those points are checked. This process is repeated until no new grid ridge points are found.
\begin{figure}
\centering

\begin{minipage}[b]{.4\linewidth}
\centering
\includegraphics[width=1 \textwidth]{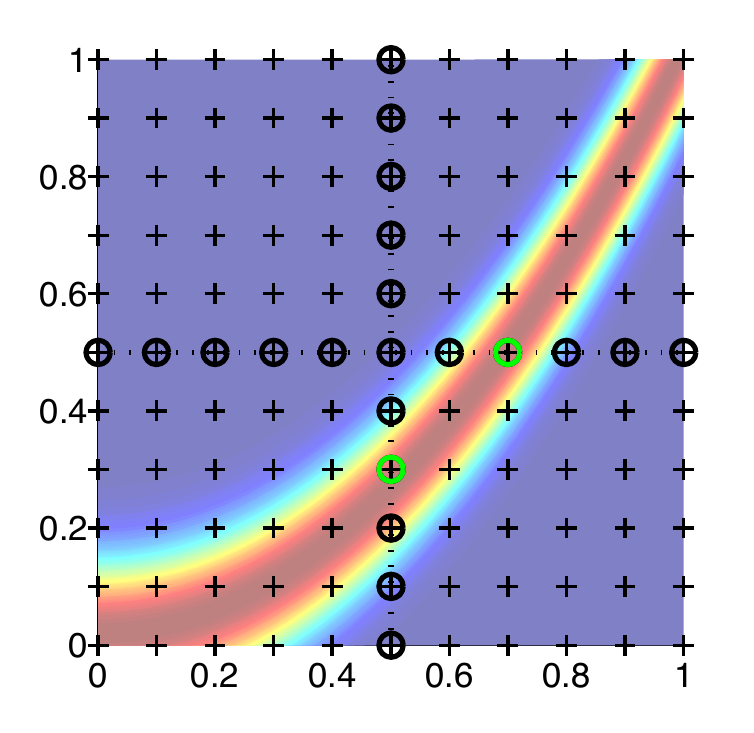}
\end{minipage}
\begin{minipage}[b]{.4\linewidth}
\centering
\includegraphics[width=1 \textwidth]{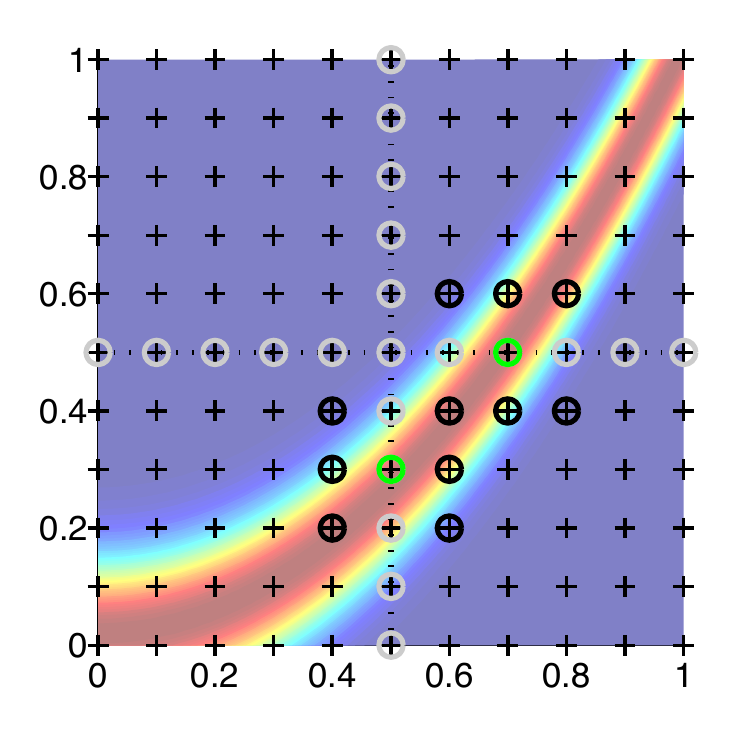}
\end{minipage}\\

\begin{minipage}[b]{.4\linewidth}
\centering
\footnotesize
{\bf (a)} Initial grid ridge points detection step, green circles are grid ridge points.
\end{minipage}
\begin{minipage}[b]{.4\linewidth}
\centering
\footnotesize
{\bf (b)} The next set of points to check for ridges (black $\mathbf \circ$'s) and previously checked points (gray $\mathbf \circ$'s).
\end{minipage}

\begin{minipage}[b]{.4\linewidth}
\centering
\includegraphics[width=1 \textwidth]{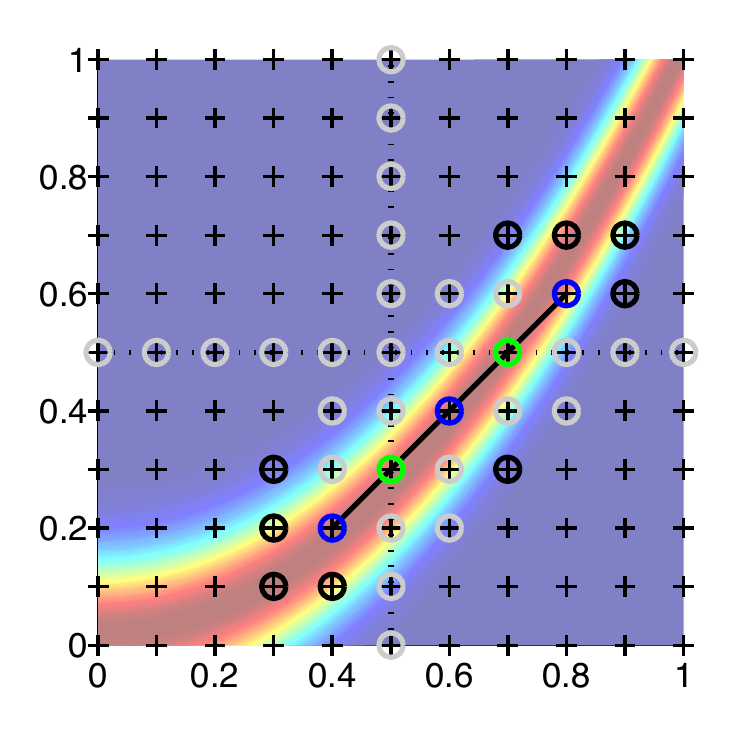}
\end{minipage}
\begin{minipage}[b]{.4\linewidth}
\centering
\includegraphics[width=1 \textwidth]{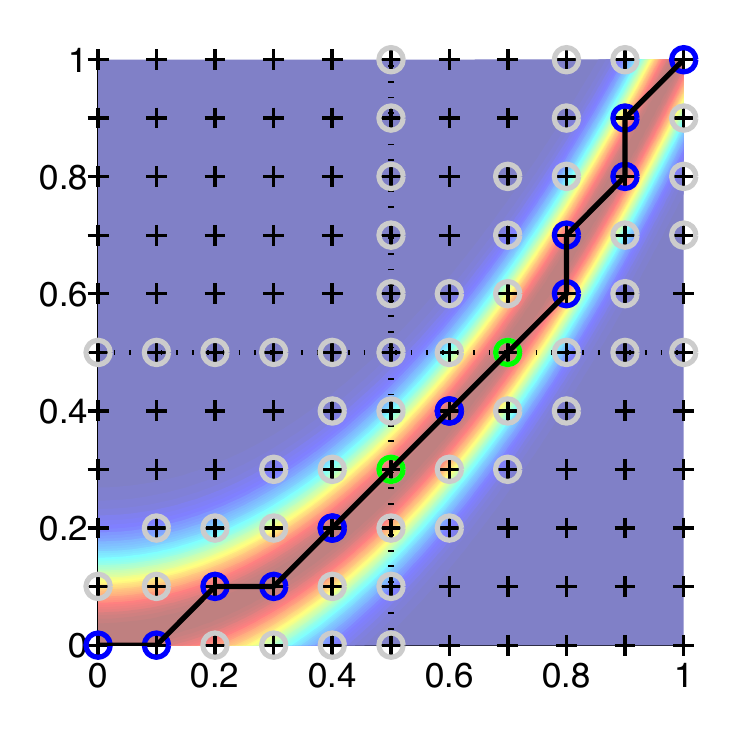}
\end{minipage}\\

\begin{minipage}[b]{.4\linewidth}
\centering
\footnotesize
{\bf (c)} New grid ridges (blue), new points to check (black) and previously checked points (gray).
\end{minipage}
\begin{minipage}[b]{.4\linewidth}
\centering
\footnotesize
{\bf (d)} The resulting LCS.
\end{minipage}
 
 \caption{The ridge tracking process in 2D. The background color represents the FTLE field, $\mathbf +$'s mark grid points. The 3D process is analogous, but a surface triangulation is used instead of line segments.}
 \label{fig:process}
\end{figure}

It is important to note that this algorithm detects {\it grid ridge points} as defined above, and therefore detects {\it coordinate local maxima} (also defined above). If the ridge height is not constant (i.e. $\nabla F\ne \mathbf{0}$ on the ridge) the grid ridge points will not converge to the actual ridges. Section \ref{sec:properties} provides an error bound and corresponding proofs addressing this issue, but the error is typically smaller than reasonable grid spacings.

\subsection{Surface triangulation}

At the end of the ridge tracking process, a large list of grid ridge points is obtained. These points must be connected into a surface triangulation for visualization and further analysis. Initial attempts at developing a 3D ridge tracking algorithm revealed that grid-less techniques requiring surface meshing where both complex and expensive because of the surface meshing process. By using a gridded coordinate system, it is possible to very quickly generate a surface triangulation from a lookup table similar to the process used in the marching cubes algorithm that is popular for isosurface generation~\cite{LorensenWE:87a}.

\begin{figure}
\centering
\subfigure[]{
\includegraphics[width = .17\textwidth]{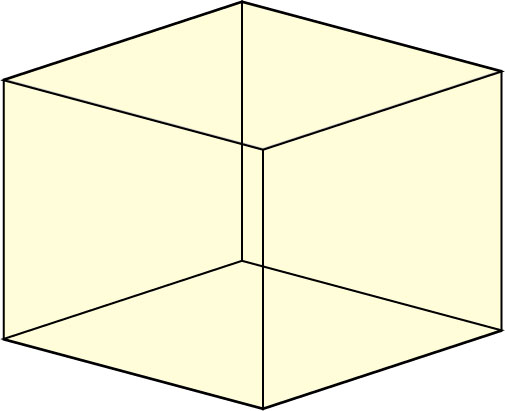}
}
\subfigure[]{
\includegraphics[width = .17\textwidth]{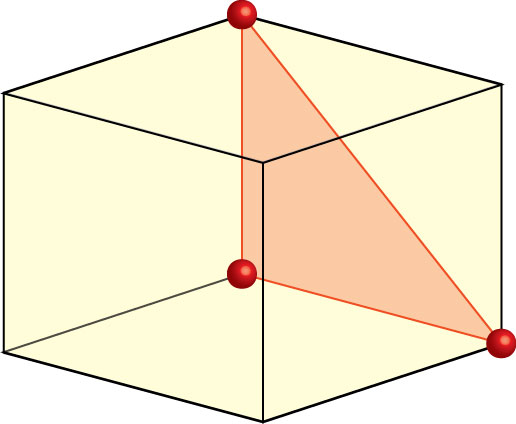}
}
\subfigure[]{
\includegraphics[width = .17\textwidth]{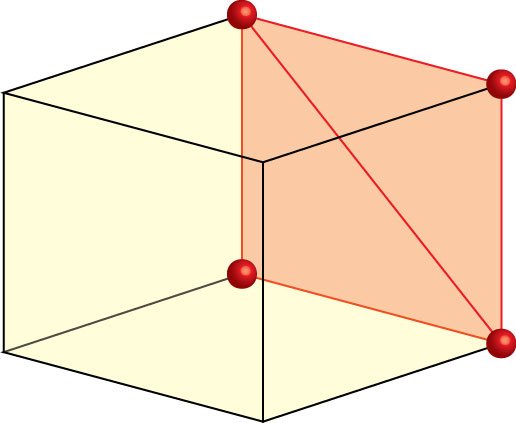}
}
\subfigure[]{
\includegraphics[width = .17\textwidth]{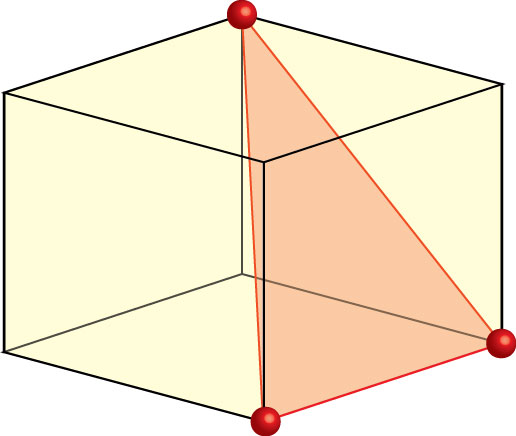}
}
\subfigure[]{
\includegraphics[width = .17\textwidth]{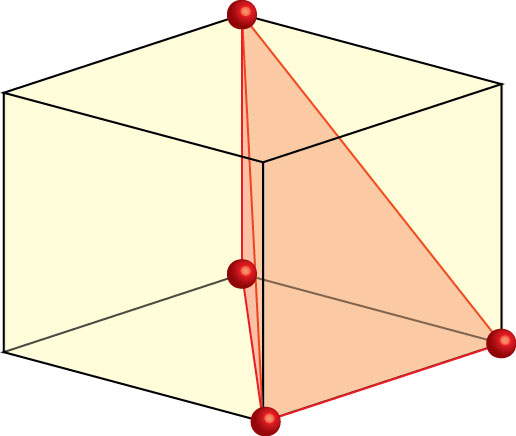}
}
\subfigure[]{
\includegraphics[width = .22\textwidth]{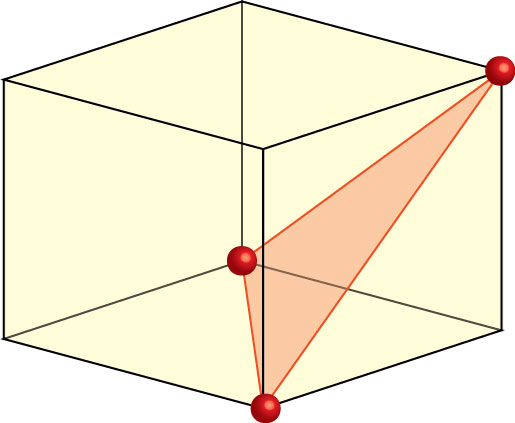}
}
\subfigure[]{
\includegraphics[width = .22\textwidth]{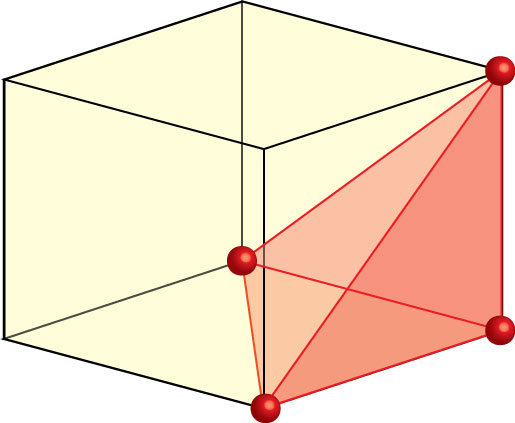}
}
\subfigure[]{
\includegraphics[width = .22\textwidth]{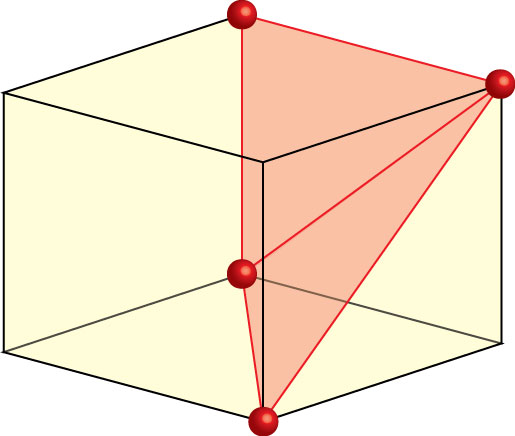}
}
\subfigure[]{
\includegraphics[width = .22\textwidth]{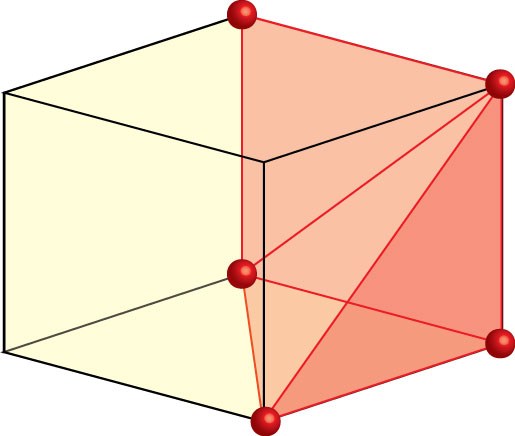}
}
\subfigure[]{
\includegraphics[width = .22\textwidth]{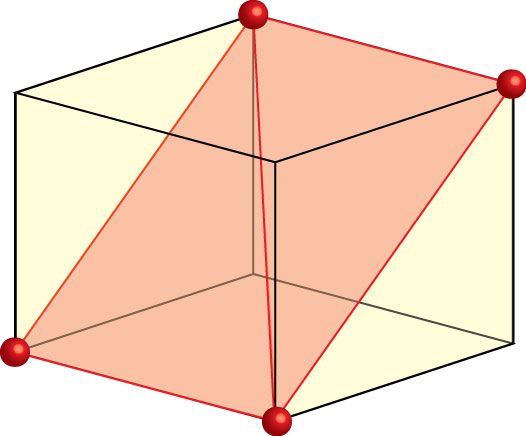}
}
\subfigure[]{
\includegraphics[width = .22\textwidth]{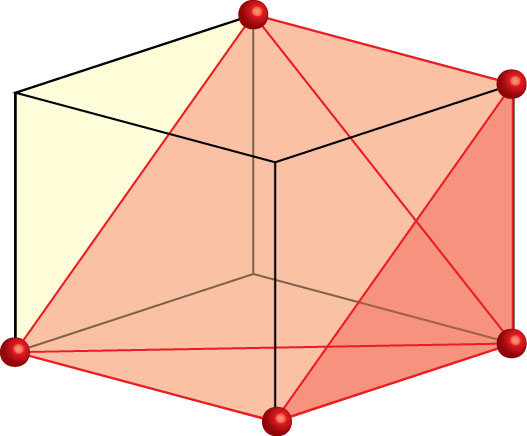}
}
\subfigure[]{
\includegraphics[width = .22\textwidth]{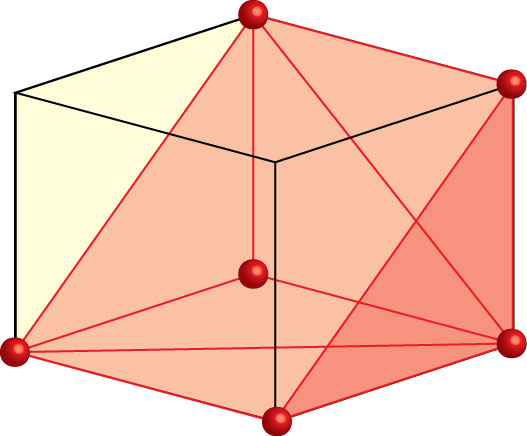}
}
\subfigure[]{
\includegraphics[width = .22\textwidth]{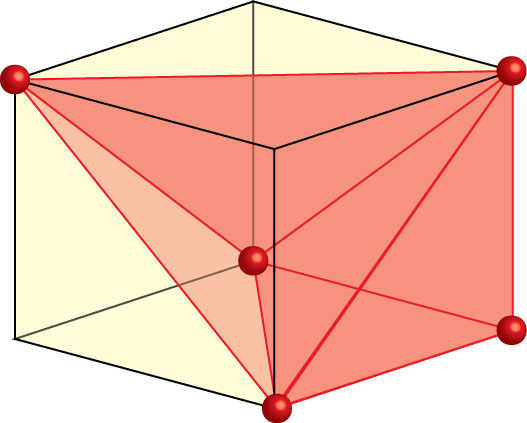}
}
\subfigure[]{
\includegraphics[width = .22\textwidth]{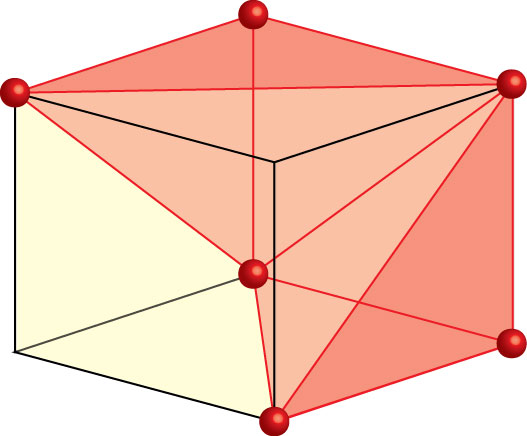}
}
\subfigure[]{
\includegraphics[width = .22\textwidth]{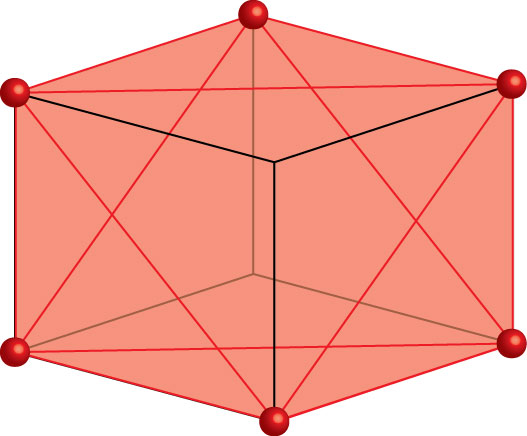}
}
\subfigure[]{
\includegraphics[width = .22\textwidth]{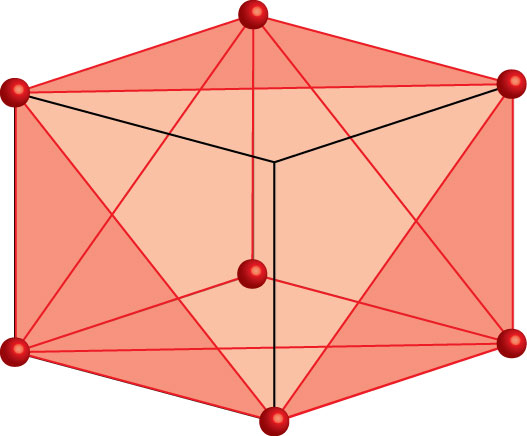}
}
\subfigure[]{
\includegraphics[width = .22\textwidth]{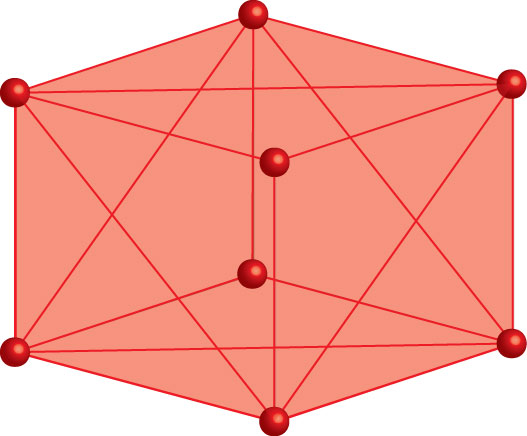}
}
 \caption{The canonical cases used in the lookup table for surface triangulation.}
 \label{fig:cases}
\end{figure}

After all the grid ridge points have been found, each cubic element in the domain is given an 8-bit number corresponding to the configuration of the grid ridge points on the elements 8 vertices. The 256 possible configurations are listed in a lookup table which directly converts the 8-bit number to a triangulation of the grid ridge points in the element.

Generating the lookup table is accomplished by using reflections and rotations to reduce the the 17 canonical cases shown in Fig. \ref{fig:cases}. Elements containing two or fewer grid ridge points contain no LCS surface triangles while elements containing 3 or more grid ridge points are added to the surface triangulation. 

The lookup table makes the surface triangulation step in this algorithm extremely efficient, but at the cost of adding some ambiguity in the precise triangulation that should be used for certain ridge point configurations. This is a well known problem in isosurface construction and is commonly dealt with by making assumptions about the field at the subgrid scale and testing additional points within the element. Similar tests may be possible for the FLTE ridges, but we have chose to err on the side of adding extra triangles to these ambiguous cases surfaces to avoid gaps in the LCS surfaces.

\section{Algorithm properties}
\label{sec:properties}

In this section we address the important properties of the ridge tracking algorithm by proving that well defined ridges are detected by the ridge tracking algorithm and the error in the location of the ridges is typically very small. The theorems contained below rely on the ridge definition presented above in Section \ref{sec:algorithm}. A few important properties are summarized here for convenience:
\begin{itemize}
\item $\nabla F$ is parallel to the ridge.
\item The eigenvectors $\{\mathbf{v}_i\}$ of the Hessian $H$ of $F$ form a complete, orthonormal basis for the space $\mathbb{R}^n$ since $H$ is symmetric.
\item The eigenvector $\mathbf{v}_1$ associated with the minimum eigenvector $\lambda_1$ of $H$ is normal to the ridge.
\end{itemize}

To prove that the ridge tracking algorithm accurately detects ridges, we first show that grid ridges converge to coordinate local maxima and then show that well defined ridges always have a nearby coordinate local maximum and find a bound on the distance between a ridge and the nearest coordinate local maximum.

%%%%%%%%%%%%%%%%%%%%%%%%%%%%%%%%%%%%%%%%%%%%%
%%%%%%%%%%%%%%%%%%%%%%%%%%%%%%%%%%%%%%%%%%%%%
%%%%%%%%%%%%%%%%%%%%%%%%%%%%%%%%%%%%%%%%%%%%%
\begin{theorem}
\label{thm:ridgeclose2grid}
Let $F$ be a $C^2$ continuous function and let $L$ be a grid line parallel to $\mathbf{e}_i$. Assume $\mathbf{x}^*$ on $L$ is a local maximum of $F$ in the $\mathbf{e}_i$ direction and the second derivative of $F$ in the $\mathbf{e}_i$ direction is negative. Then there exists a value $\epsilon$ such that if the grid spacing $dx<\epsilon\implies \norm{\mathbf{x}^*-\mathbf{x}^g}\le dx$ for some grid point $\mathbf{x}^g$
\end{theorem}
\begin{proof}
Since $\mathbf{x}^*$ is a local maximum in the $\mathbf{e}_i$ direction, there is a value $\varepsilon_1>0$ such that for all $\abs{c_1}<\varepsilon_1$, $F|_{\mathbf{x}^*+c_1\mathbf{e}_i}<F|{\mathbf{x}^*}$. Since $F$ is $C^2$ continuous and $\partial^2 F/\partial {\mathbf{e}_i}^2<0$, $\exists\ \varepsilon \in (0,\varepsilon_1)$ such that 
\begin{equation}
\left.\dfrac{\partial^2 F}{\partial {\mathbf{e}_i}^2}\right|_{\mathbf{x}^*+c\mathbf{e}_1}<0\ \forall\ \abs{c}<\varepsilon.
\end{equation}
This implies that $\partial F/\partial {\mathbf{e}_i}$ is monotonically decreasing over the interval $c\in (-\varepsilon, \varepsilon)$.

Next, choose a grid spacing, $dx<\varepsilon=\varepsilon/2$. This ensures that at least five grid points lie in the intervalal $(\mathbf{x}^*-\varepsilon \mathbf{e}_i,\mathbf{x}^*+\varepsilon \mathbf{e}_i)$ and at least lie two to each side of $\mathbf{x}^*$.

If $\mathbf{x}^*$ is a grid point, then the neighboring grid points have lower values of $F$ since $dx<\varepsilon<\varepsilon_1$ so this grid point is a grid ridge point and there is a grid ridge point less than dx from ${\mathbf{x}^*}$.

If $\mathbf{x}^*$ is not a grid point, label the nearest four grid points $\mathbf{x}_1$, $\mathbf{x}$, $\mathbf{x}_3$, $\mathbf{x}_4$, with $\mathbf{x}^*$ lying between $\mathbf{x}$ and $\mathbf{x}_3$. Since $\partial F/\partial {\mathbf{e}_i}$ is monotonically decreasing over the interval and $\partial F/\partial {\mathbf{e}_i}|_{\mathbf{x}^*}=0$, $F|_{\mathbf{x}_1}<F|_{\mathbf{x}}$ and $F|_{\mathbf{x}_3}>F|_{\mathbf{x}_4}$. If $F|_{\mathbf{x}}>F|_{\mathbf{x}_3}$, $\mathbf{x}$ is a grid ridge point. If $F|_{\mathbf{x}}<F|_{\mathbf{x}_3}$, $\mathbf{x}_3$ is a grid ridge point. If $F|_{\mathbf{x}}=F|_{\mathbf{x}_3}$, both $\mathbf{x}$ and $\mathbf{x}_3$ are grid ridge points. In any of these cases, since $\mathbf{x}^*$ lies between $\mathbf{x}$ and $\mathbf{x}_3$, there is a grid ridge point less than $dx$ from $\mathbf{x}^*$.
\end{proof}
%%%%%%%%%%%%%%%%%%%%%%%%%%%%%%%%%%%%%%%%%%%%%
%%%%%%%%%%%%%%%%%%%%%%%%%%%%%%%%%%%%%%%%%%%%%
%%%%%%%%%%%%%%%%%%%%%%%%%%%%%%%%%%%%%%%%%%%%%

Thm. \ref{thm:ridgeclose2grid} proves that as grid spacing goes to zero, there are grid ridge points that converge to the coordinate local maxima of $F$. Therefore the ridge tracking algorithm detects coordinate local maxima. We now show that every well defined ridge has a nearby coordinate local maximum.

\begin{theorem}
\label{thm:error}
Given a $C^3$ continuous function $F$ that admits a sufficiently sharp second derivative ridge, for every point $\mathbf{x}_0$ on the ridge, there is a nearby point, $\mathbf{x}^*$, in the ridge normal direction that is a coordinate local maximum. The nearest coordinate local maximum is no further from the ridge than 
\[
d=\dfrac{2\sqrt{n-1}\norm{ \mathbf{D}F(\mathbf{x}_0) }}{\abs{\lambda_1}},
\]
as long as 
\[
\norm{ \mathbf{v}_1\cdot(\mathbf{D}^3F) }< \frac{1}{2}\frac{ \lambda_1^2 }{n\sqrt{n-1} \norm{ \mathbf{D}F(\mathbf{x}_0) }}
\]
and
\[
\lambda_n \le \dfrac{n}{n-1} \abs{ \dfrac{\lambda_1}{n} + \dfrac{2\sqrt{n-1}}{\abs{\lambda_1}} \norm{\mathbf{D}F(\mathbf{x}_0)} \norm{\mathbf{v}_1\cdot\mathbf{D}F(\mathbf{\xi}) } },
\]
where $\mathbf{D}$ denotes the gradient operator, $\lambda_1$ (respectively $\lambda_n$) is the minimum (respectively maximum) eigenvalue of the Hessian, $\mathbf{D}^2 F(\mathbf{x}_0)$, $\mathbf{v}_1$ and $\mathbf{v}_n$ are the eigenvectors associated with $\lambda_1$ and $\lambda_n$, and $n$ is the dimension of the space.
\end{theorem}

Although these criteria may seem restrictive, they are typically easily satisfied by well defined ridges as we will see in examples below. $\lambda_1$ is typically very large in magnitude (sometimes up to $10^{12}$) while $\norm{\mathbf{D}F}$ and $\abs{\lambda_n}$ are typically $\mathcal{O}(1)$. If $\lambda_n$ is negative the third condition is trivially satisfied.

\begin{proof}
By assumption $F$ is $C^3$ continuous and admits a second derivative ridge through point $\mathbf{x}_0$. Denote the eigenvalues and normalized eigenvectors of the Hessian $\mathbf{D}^2 F(\mathbf{x}_0)$ as $\lambda_1\le \lambda_2 \le  \dotsb \le \lambda_n$ and $\mathbf{v}_1,\mathbf{v}, \dotsc ,\mathbf{v}_n$. By definition, $\mathbf{v_1}$ is normal to the ridge and $\mathbf{D}F(\mathbf{x}_0)$ is parallel to the ridge. We also know that $\{\mathbf{v}_1, \dotsc ,\mathbf{v}_n\}$ forms an orthonormal basis for the space since the Hessian is a real symmetric matrix. Also, let the coordinate system for the space be defined by the orthonormal basis vectors $\{\mathbf{e}_1, \dotsc ,\mathbf{e}_n\}$.

We first choose the coordinate direction that is closest to the ridge normal direction. That is, we choose $\mathbf{e}_j$ from the set $\{\mathbf{e}_i\}$ such that $\abs{ \mathbf{e}_j\cdot\mathbf{v}_1 }\ge \abs{ \mathbf{e}_i\cdot\mathbf{v}_1 } \ \forall \ i \in\{1, \dotsc ,n\}$. By Thm. \ref{thm:closest_vector} we know that
\begin{equation}
\abs{ \mathbf{e}_j\cdot\mathbf{v}_1 }\ge 1/\sqrt{n}. \label{eq:dot_limit}
\end{equation}

The gradient of $F$ near $\mathbf{x}_0$ may be written as
\[
\mathbf{D}F(\mathbf{x}) = \mathbf{D}F(\mathbf{x}_0)+\left((\mathbf{x}-\mathbf{x}_0)\cdot \mathbf{D}\right) \mathbf{D}F(\mathbf{x}_0)+\frac{1}{2!}\left((\mathbf{x}-\mathbf{x}_0)\cdot\mathbf{D}\right)^2 \mathbf{D}F(\mathbf{\xi})
\]
for some $\mathbf{\xi}$ that is a linear combination of $\mathbf{x}$ and $\mathbf{x}_0$. Taking the dot product of $\mathbf{D}F(\mathbf{x})$ with $\mathbf{e}_j$ and considering only points on a line normal to the ridge such that $\mathbf{x}=\mathbf{x}_0+d\mathbf{v}_1$ gives
\[
\mathbf{e}_j \cdot \mathbf{D}F(\mathbf{x}) = \mathbf{e}_j \cdot \mathbf{D}F(\mathbf{x}_0)+d \mathbf{e}_j \cdot \left\{(\mathbf{v}_1 \cdot \mathbf{D}) \mathbf{D}F(\mathbf{x}_0)\right\}+\frac{1}{2!}d^2\mathbf{e}_j\cdot\left\{(\mathbf{v}_1 \cdot \mathbf{D})^2 \mathbf{D}F(\mathbf{\xi})\right\}.
\]
Since $\mathbf{v}_1$ is an eigenvalue of $D^2F(\mathbf{x}_0)$, this can be rewritten as
\begin{equation}
\label{eq:quadratic}
\mathbf{e}_j \cdot \mathbf{D}F(\mathbf{x}) = \mathbf{e}_j \cdot\mathbf{D}F(\mathbf{x}_0)+d\lambda_1 \mathbf{e}_j \cdot \mathbf{v}_1+\frac{1}{2!}d^2\mathbf{e}_j\cdot(\mathbf{v}_1 \cdot \left\{\mathbf{D}^3F(\mathbf{\xi})\right\})
\end{equation}
We can establish an upper bound on $\mathbf{e}_j \cdot\mathbf{D}F(\mathbf{x}_0)$ by noting that $\mathbf{v}_1\cdot\mathbf{D}F(\mathbf{x}_0)=0$ and expressing $\mathbf{e}_j$ and $\mathbf{D}F(\mathbf{x}_0)$ in terms of $\{\mathbf{v}_i\}$:
\begin{align}
\abs{ \mathbf{e}_j \cdot\mathbf{D}F(\mathbf{x}_0) } & = \left( \sum_{i=1}^n (\mathbf{v}_i\cdot\mathbf{e}_j)\mathbf{v}_i\right) \cdot \left( \sum_{i=2}^n (\mathbf{v}_i\cdot\mathbf{D}F(\mathbf{x}_0))\mathbf{v}_i \right) \notag \\
 & = \left( \sum_{i=2}^n (\mathbf{v}_i\cdot\mathbf{e}_j)\mathbf{v}_i\right) \cdot \left( \sum_{i=2}^n (\mathbf{v}_i\cdot\mathbf{D}F(\mathbf{x}_0))\mathbf{v}_i \right) \notag \\
 & \le \norm{ \sum_{i=2}^n (\mathbf{v}_i\cdot\mathbf{e}_j)\mathbf{v}_i }\ \norm{ \sum_{i=2}^n (\mathbf{v}_i\cdot\mathbf{D}F(\mathbf{x}_0))\mathbf{v}_i } \notag \\
  & \le \sqrt{\sum_{i=2}^n (\mathbf{v}_i\cdot\mathbf{e}_j)^2}  \sqrt{\sum_{i=2}^n (\mathbf{v}_i\cdot\mathbf{D}F(\mathbf{x}_0))^2} \notag \\
  & \le \sqrt{1-\frac{1}{n}} \norm{ \mathbf{D}F(\mathbf{x}_0) } \label{eq:limit_1}.
\end{align}

The right hand side of Eq. \ref{eq:quadratic} is a quadratic function in $d$ with bounds on the coefficients
\begin{align*}
 0 \le & \abs{c} \le c_1=\sqrt{1-1/n}\norm{\mathbf{D}F(\mathbf{x}_0)}, \\
 0 < b_0=\frac{\abs{\lambda_1}}{\sqrt{n}} \le & \abs{b} \le b_1=\abs{\lambda_1}, \\
 0 \le & \abs{a} \le \dfrac{\lambda_1^2}{4n\sqrt{n-1}\norm{\mathbf{D}F(\mathbf{x}_0)}} < \frac{b_0^2}{4c_1}.
\end{align*}
Therefore, by Thm. \ref{thm:quadratic}, this function has a root at $d^*$ with
$$
\abs{d^*} < \dfrac{2\sqrt{n-1}\norm{ \mathbf{D}F(\mathbf{x}_0) }}{\lambda_1}.
$$
where $\mathbf{e}_j \cdot \mathbf{D}F(\mathbf{x}) = 0$.

We have now shown that the value of $\mathbf{e}_j \cdot \mathbf{D}F(\mathbf{x})$ is zero somewhere on the line normal to the ridge at $\mathbf{x}_0$ at a distance of less than 
\[2\sqrt{(1-1/n)} \norm{ \mathbf{DF(\mathbf{x}_0)} }/\lambda_1.\]
 Call the zero point $\mathbf{x}^*$. To complete the proof, we will show that $\mathbf{e}_j^T (\mathbf{D}^2F) \mathbf{e}_j$ must be negative at this point so it is a coordinate local maximum. The first order Taylor series for the Hessian near $\mathbf{x}_0$ is
\[
\mathbf{D}^2 F(\mathbf{x}) = \mathbf{D}^2 F(\mathbf{x}_0) + ((\mathbf{x}-\mathbf{x}_0) \cdot \mathbf{D})\mathbf{D}^2F(\mathbf{\xi}).
\]
At $\mathbf{x}^*$ the second derivative of $F$ in the $\mathbf{e}_j$ direction is given by
\begin{align*}
\mathbf{e}_j^T (\mathbf{D}^2 F(\mathbf{x}^*)) \mathbf{e}_j = \mathbf{e}_j^T (\mathbf{D}^2 F(\mathbf{x}_0)) \mathbf{e}_j+ d^* \mathbf{e}_j^T (\mathbf{v}_1 \cdot \mathbf{D}^3F(\mathbf{\xi})) \mathbf{e}_j
\end{align*}
where $d^* < 2\sqrt{n-1}\norm{ \mathbf{D}F(\mathbf{x}_0) }/\lambda_1$ is the distance from the ridge of the point where $\mathbf{e}_j \cdot \mathbf{D}F(\mathbf{x})=0$. Then
\begin{align*}
\mathbf{e}_j^2 \mathbf{D}^2 F(\mathbf{x}) & < \dfrac{\lambda_1}{n} + \left(1-\dfrac{1}{n} \right) \lambda_n + \dfrac{2\sqrt{n-1}\norm{ \mathbf{D}F(\mathbf{x}_0) }}{\abs{ \lambda_1 }} \norm{\mathbf{v}_1 \cdot \mathbf{D}^3F(\mathbf{\xi})} \\
% & < \dfrac{\lambda_1}{2n} + \left(1-\dfrac{1}{n} \right) \lambda_n \\
 & < 0
\end{align*}
where $\norm{\mathbf{v}_1 \cdot \mathbf{D}^3F(\mathbf{\xi})}$ is understood as the induced norm of the matrix that results from the tensor dot product $\mathbf{v}_1 \cdot \mathbf{D}^3F(\mathbf{\xi})$.
Thus, at $\mathbf{x}^*$ the first derivative of $F$ is zero in the $\mathbf{e}_j$ direction and the second derivative is negative. $\mathbf{x}^*$ is a coordinate local maximum.

\end{proof}
As long the criteria of Thm. \ref{thm:ridgeclose2grid} are satisfied along a ridge, there is a nearby coordinate local maximum that will be detected by the ridge tracking algorithm. The error of the location of the detected ridges is bounded by 
\[
d=\dfrac{2\sqrt{n-1}\norm{ \mathbf{D}F(\mathbf{x}_0) }}{\abs{\lambda_1}}.
\]
Since $\norm{\mathbf{D}F}$ is typically $\mathcal{O}(1)$ (if the ridge does not rise and fall too quickly), the error is typically $\mathcal{O}(1/\abs{\lambda_1})$.

\subsection{Ridge examples}
\label{sec:ridge_examples}

We now verify the properties of the ridge tracking algorithm on two two-dimensional examples. An analytically defined surface that admits a ridge and the FTLE field for a time dependent double gyre flow. We begin by examining the surface
\begin{equation}
\label{eq:func1}
F(x,y) = \dfrac{e^{-0.1(x-y)^2}}{\ln((x+y)^2+15)}.
\end{equation}

\begin{wrapfigure}{r}{2.5in}
 \includegraphics[width=2.4in]{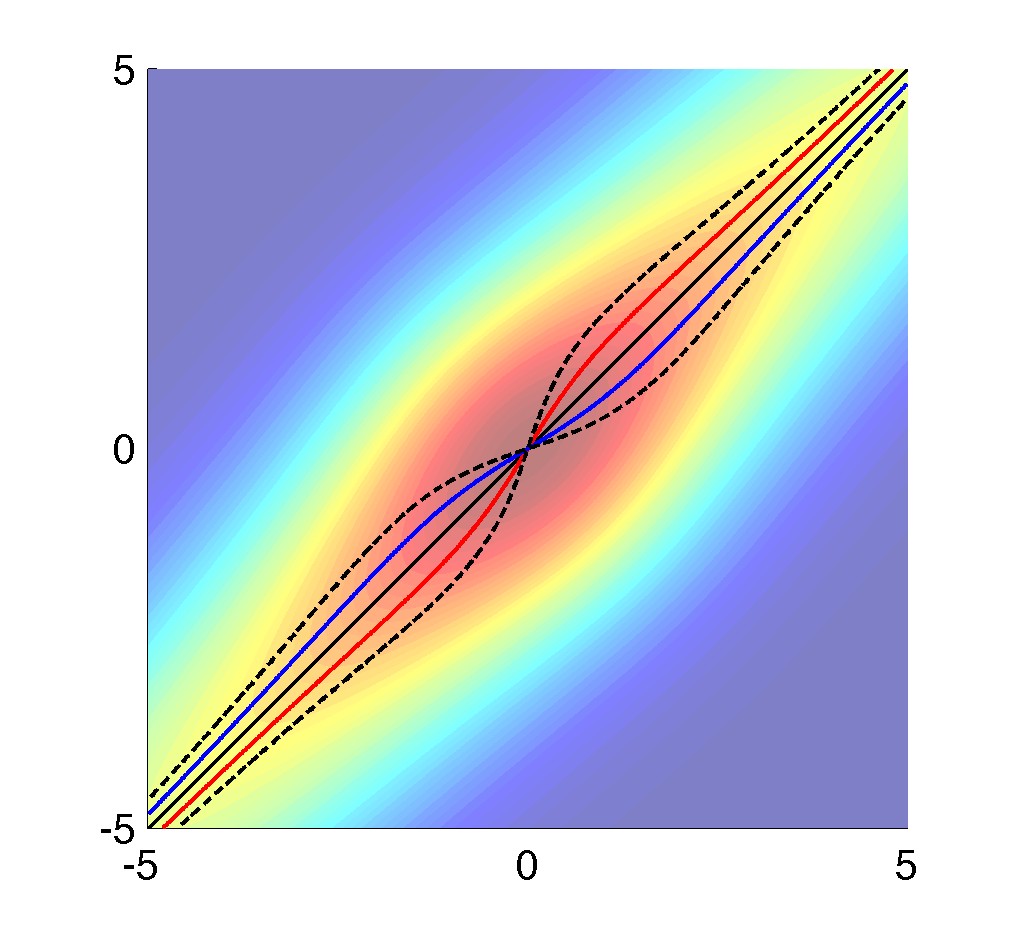}
 \caption{The surface defined by Eq. \ref{eq:func1}. The solid black line denotes the ridge while the dashed black curves show the error bound and the red and blue curves show the coordinate local maxima.}
\label{fig:func1}
\end{wrapfigure}
A contour plot of this function is shown in Fig. \ref{fig:func1} and the function admits a ridge along the line $y=x$. This represents a worst case scenario in terms of ridge orientation since the ridge is at a $45^\circ$ angle to both coordinate directions. The ridge is also much thicker than is typically seen in FTLE ridges. Since this ridge is defined by an analytical function, it is possible to easily compute all the necessary criteria for Thm. \ref{thm:error}. In Fig. \ref{fig:func1}, the actual ridge has been drawn as a solid black line. The error bound, $d$ is drawn in the figure as a pair of dashed black lines and the the coordinate local maxima have been drawn as red (for $\partial F/\partial x=0$) and blue (for $\partial F/\partial y=0$) lines. The other criteria of Thm. \ref{thm:error} are all satisfied despite $\lambda_1$ being much larger than is usually seen in FTLE ridges. We find $\lambda_1\in[-0.1477,-0.0843]$, $\lambda_2\in[-0.0363,0.0071]$, and $\norm{\mathbf{D}F}\le0.0356$ on the ridge in the domain $(x,y)\in[-5,5]^2$.

As a second example to apply the above theorem, we consider a ridge in the FTLE field of a time dependent double gyre. The velocity field consists of two counter rotating gyres with a periodic perturbation that enables transport between the two gyres. The flow is given by the stream function
\begin{equation}
 \psi(x,y,t) = A\sin(\pi f(x,t)))\sin(\pi y) 
\end{equation}
on the domain $[0,2]\times[0,1]$ where
\begin{align}
 f(x,t) & = a(t)x^2+b(t), \notag \\
 a(t) & = \epsilon \sin(\omega t), \\
 b(t) & = 1-2\epsilon\sin(\omega t). \notag
\end{align}
The velocity is given by
\begin{equation}
 u = -\frac{\partial \psi}{\partial y},\ \ v = \frac{\partial \psi}{\partial x}.
\end{equation}
We use parameters $A=0.1$, $\epsilon=0.25$, and $\omega=10$ for this example and set the integration time for computing the FTLE at $T=15$. We will only consider the forward time FTLE field at time $t=0$ which is shown in Fig. \ref{fig:dg2d}.
\begin{figure}
\begin{minipage}[b]{0.49\linewidth}
\centering
 \includegraphics[width=1\linewidth]{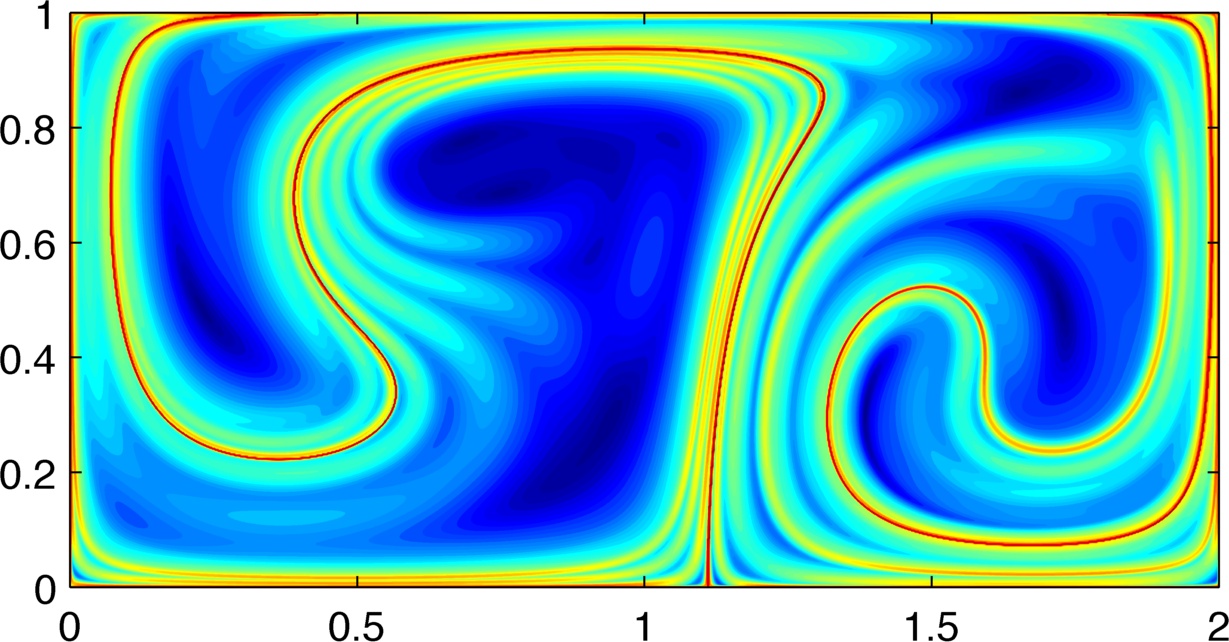}
\end{minipage}
\begin{minipage}[b]{0.49\linewidth}
\centering
 \includegraphics[width=1\linewidth]{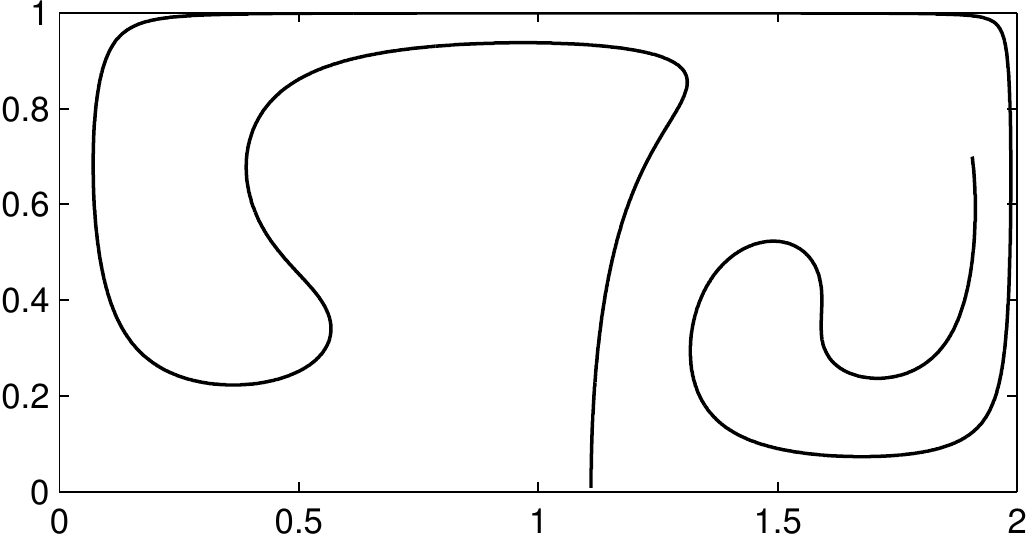}
\end{minipage}
\begin{minipage}[b]{0.49\linewidth}
\centering
 \footnotesize
\bf (a)
\end{minipage}
\begin{minipage}[b]{0.49\linewidth}
\centering
 \footnotesize
\bf (b)
\end{minipage}
 \caption{(a) The forward FTLE field and (b) the main FTLE ridge (with FTLE value $>0.25$) for the time dependent double gyre with $A=0.1$, $\epsilon=0.25$, $\omega=10$, $t=0$, and $T=15$.}
\label{fig:dg2d}
\end{figure}

The main FTLE ridge in this double gyre flow (shown in Fig. \ref{fig:dg2d}b) was computed with very high precision by iteratively estimating the ridge position and tangent direction and then adjusting the position in the normal direction. This is necessary to accurately compute the FTLE values on the ridge as well as the gradient, and Hessian of the FTLE field to bound the error in the ridge locations. Note that the ridge seen in this FTLE field is much sharper than the analytical ridge investigated in the previous example.

\begin{figure}
\begin{minipage}[b]{0.49\linewidth}
\centering
 \includegraphics[width=.95\linewidth]{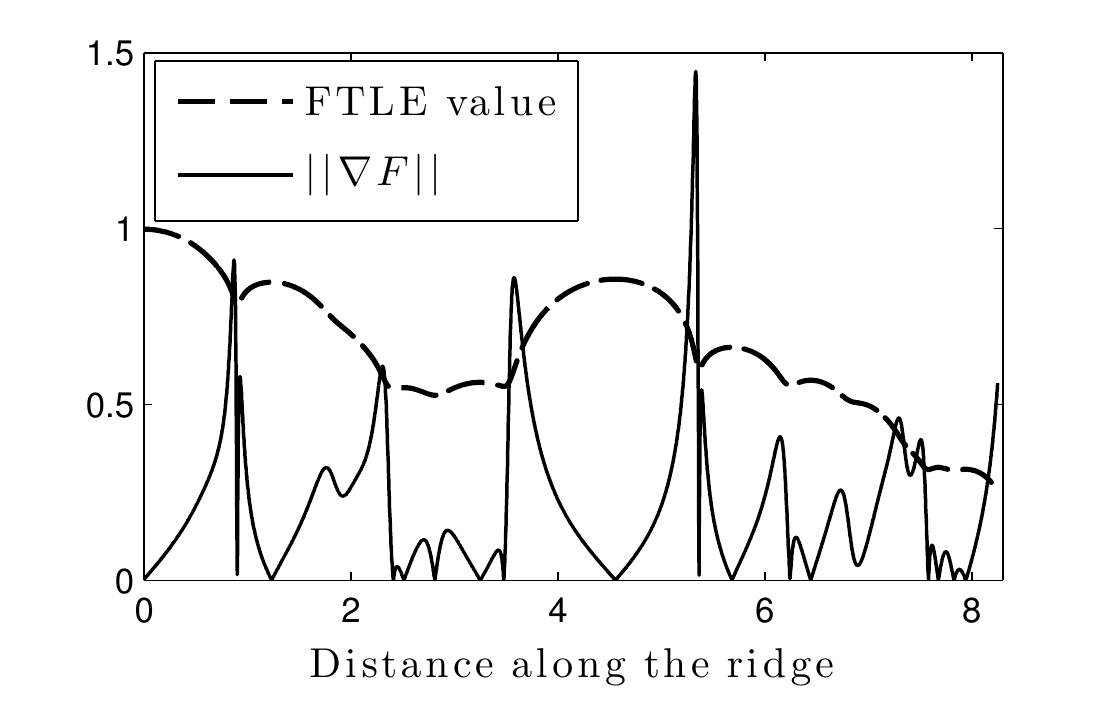}
\end{minipage}
\begin{minipage}[b]{0.49\linewidth}
\centering
 \includegraphics[width=1\linewidth]{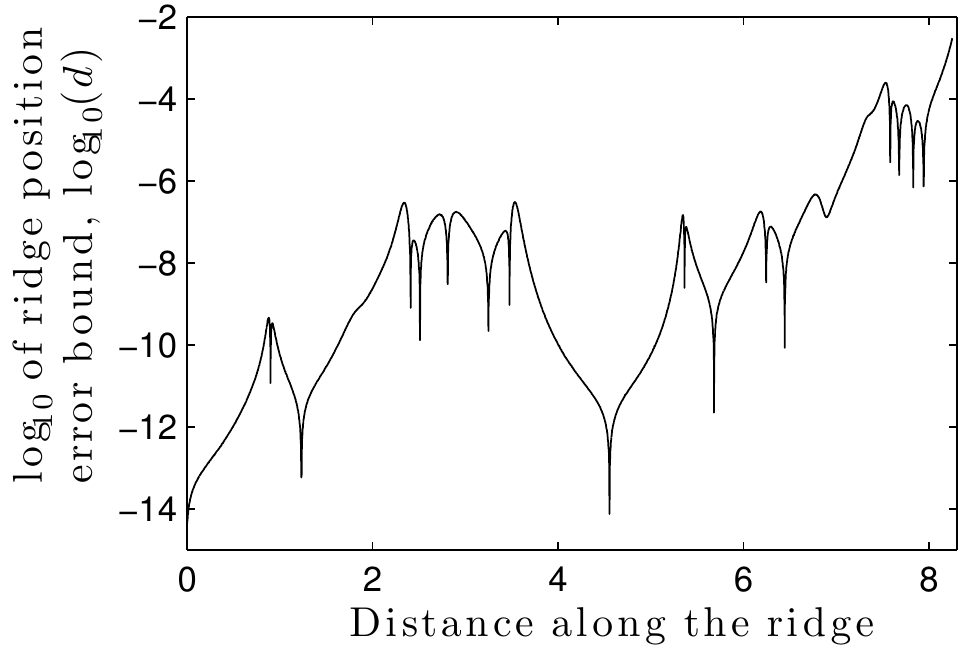}
\end{minipage}
\begin{minipage}[b]{0.49\linewidth}
\centering
 \footnotesize
\bf (a)
\end{minipage}
\begin{minipage}[b]{0.49\linewidth}
\centering
 \footnotesize
\bf (b)
\end{minipage}
 \caption{(a) The FTLE values and gradient along the ridge shown in Fig. \ref{fig:dg2d}b and (b) bound on the ridge position error as given in Thm. \ref{thm:error}. The distance along the ridge ($x$-axis) is computed as the distance along the ridge from the ridge origin at $(1.110,0)$. Since $\norm{\nabla F}=\mathcal{O}(1)$ and $\abs{\lambda_1}\gg1$, the error bound is very small, typically less than $10^{-6}$.}
\label{fig:dg2d_error}
\end{figure}
Fig. \ref{fig:dg2d_error} shows the FTLE values and the norm of the gradient of the FTLE field along the ridge as well as the bound on the ridge location error which is less than $10^{-6}$ for the majority of the ridge and never rises above $10^{-2}$. The norm of the gradient is bounded by $\norm{\nabla F}<1.5$. Additionally the eigenvalues of the Hessian fall in the range $-1.8\times 10^{12}<\lambda_1<-367$ and $-6.1<\lambda_2<37.3$ with averages of $\overline{\lambda_1}=-8.1\times 10^{10}$ and $\overline{\lambda_2}=-0.033$ and medians $\text{med}(\lambda_1)=-5.71\times 10^7$ and $\text{med}(\lambda_2)=-0.30$. The very large magnitude of $\lambda_1$ means that the grid based ridge tracking algorithm will be extremely accurate for this example. For reference, grid spacings of $10^{-2}$ or $10^{-3}$ are typically used when performing LCS computations for this problem.

\section{Algorithm performance}

In this section we discuss the performance of the ridge tracking algorithm by examining two analytically defined examples: a time dependent double gyre and Arnold-Beltrami-Childress flow. We verify that the expected LCS surfaces are extracted and then focus on establishing the computational order of the surface tracking algorithm and compare this to the standard FTLE algorithm that computes the FTLE field everywhere in the domain. 

\subsection{Time dependent double gyre}

The first example presented is the time dependent double gyre used above in section \ref{sec:ridge_examples}. We extend this flow to three dimensions by simply setting the velocity in the z direction to $w=0$. Since there is no $z$ dependence and no velocity in the $z$ direction, the LCS will be independent of z as well. We use the parameters $A = 0.1$, $\epsilon = 0.1$, and $\omega = 2\pi /10$. We set the integration time to be $T=\pm 15$ and compute both the forward and backward LCS. A threshold of $80\%$ of the maximum FTLE value is used to determine the LCS.

\begin{figure}
\begin{minipage}[b]{0.49\linewidth}
\centering
 \includegraphics[width=1 \linewidth]{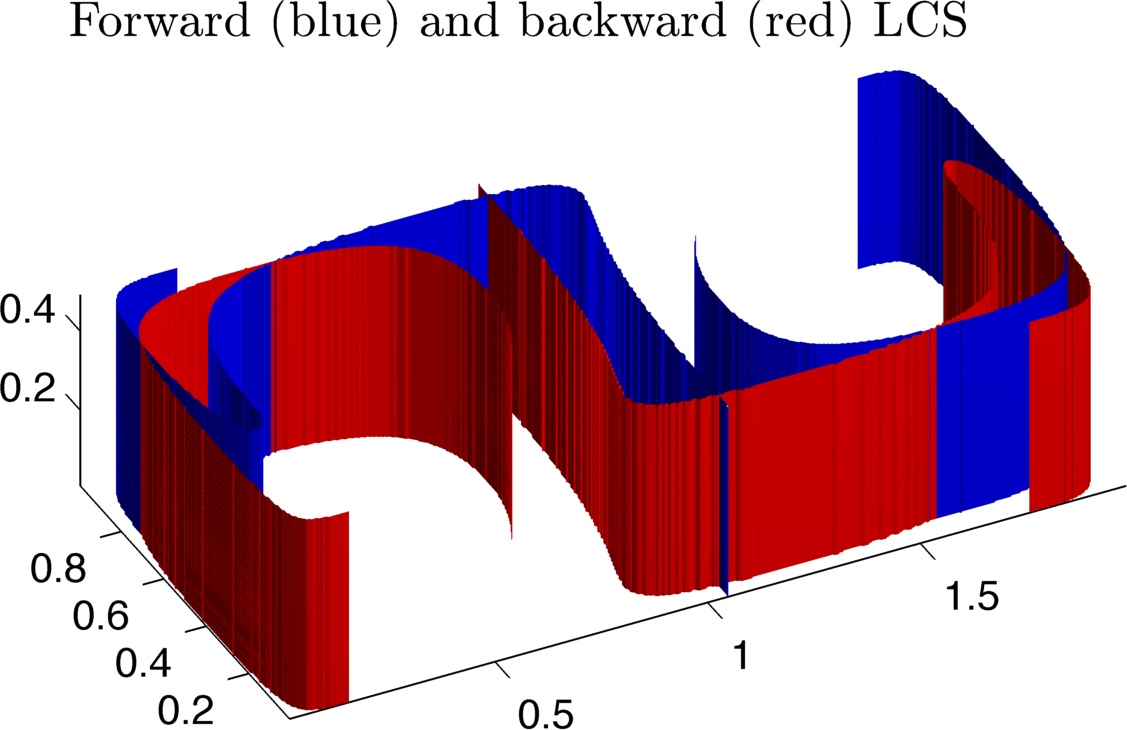}
 \end{minipage}
\begin{minipage}[b]{0.49\linewidth}
\centering
 \includegraphics[width=1\linewidth]{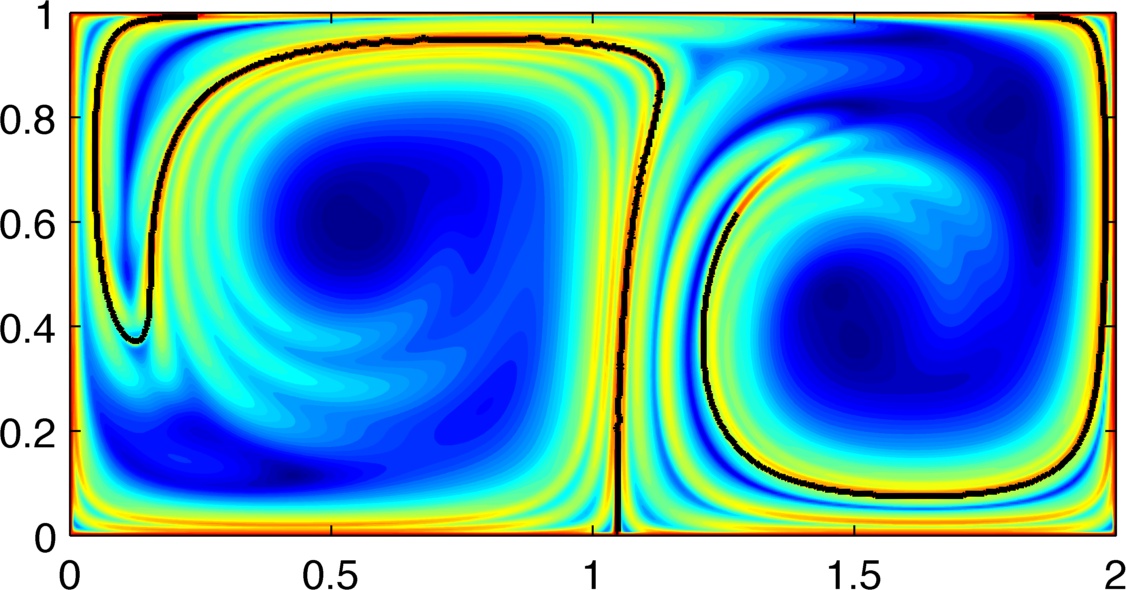}
 \end{minipage}
\begin{minipage}[b]{0.49\linewidth}
\centering
\bf (a)
\end{minipage}
\begin{minipage}[b]{0.49\linewidth}
\centering
\bf (b)
\end{minipage}
 \caption{Surface tracking results for the double gyre flow. The 3D LCS are shown at in (a) with forward LCS colored blue and backward LCS colored red. The forward FTLE field is shown in (b) with the ridge tracking results overlaid as the black curve that precisely lines up with the FTLE ridge.}
 \label{fig:dg}
\end{figure}

The LCS in this system have been well studied in the past and our results agree with previous publications~\cite{Marsden:05g,Mohseni:10g}. The full 3D LCS surfaces are shown in Fig. \ref{fig:dg}. This figure also shows the backward FTLE field overlaid with the results of the 3D ridge tracking algorithm. The LCS extracted by the ridge tracking algorithm lie exactly on top of the ridge in the FTLE field.

The computational timing results are summarized below in Section \ref{sec:timing_results}

\subsection{Arnold-Beltrami-Childress flow}

Arnold-Beltrami-Childress (ABC) flow is a three-dimensional, $2\pi$-periodic flow that has been previously studied with LCS techniques~\cite{HallerG:01a}. The flow is given by
\begin{align}
u & = A\sin(z) + C\cos(y), \notag \\
v & = B\sin(x) + A\cos(z), \\
w & = C\sin(y) + B\cos(x), \notag 
\end{align}
where $A = 1$, $B=\sqrt{2/3}$, and $C=\sqrt{1/3}$. We use an integration time of $T=10$ for FTLE computations and us a threshold value of $70\%$ of the maximum FTLE value.

\begin{figure}
\centering
\begin{minipage}[b]{0.49\linewidth}
\centering
 \includegraphics[height=2.2in]{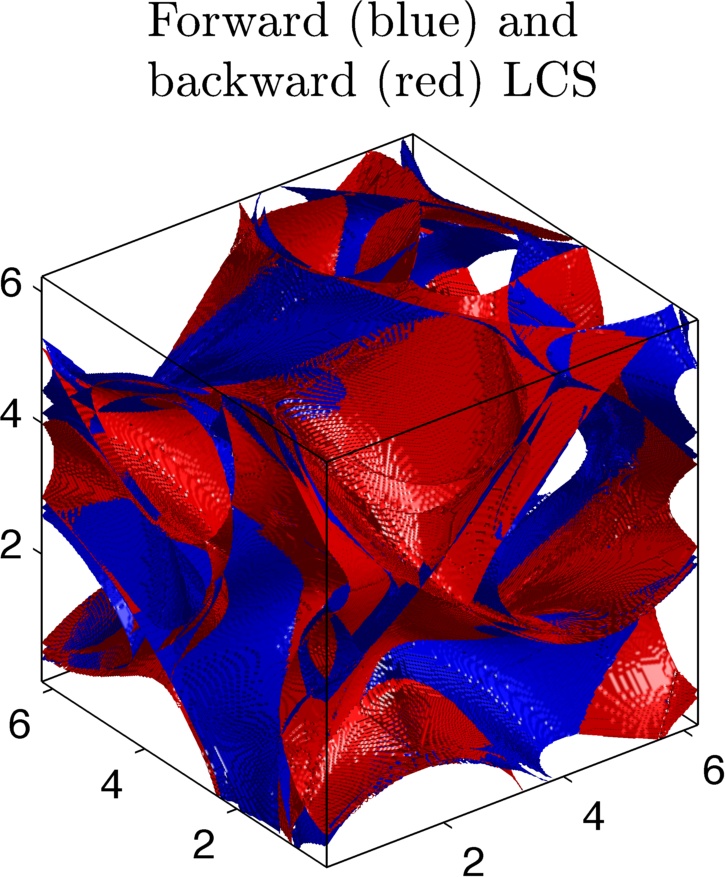}
\end{minipage}
\begin{minipage}[b]{0.49\linewidth}
\centering
 \includegraphics[height=2.2in]{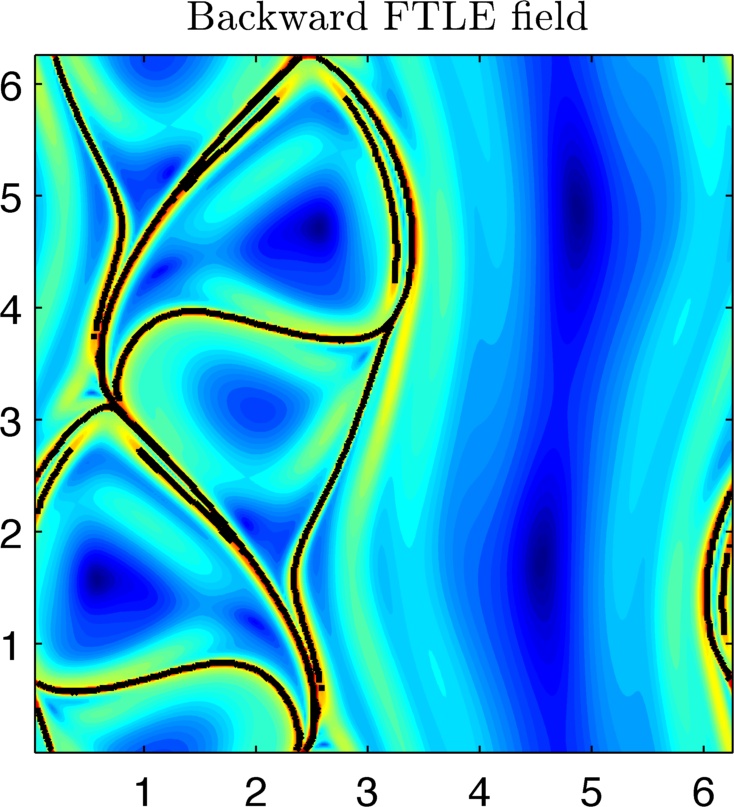}
\end{minipage}
\begin{minipage}[b]{0.49\linewidth}
\centering
\bf (a)
\end{minipage}
\begin{minipage}[b]{0.49\linewidth}
\centering
\bf (b)
\end{minipage}
 \caption{Ridge tracking results for ABC flow. The 3D LCS are shown in (a) with forward LCS colored blue and backward LCS colored red. (b) shows a single plane at $z=\pi$ to display the FTLE field (colored) and corresponding LCS (black curves) as computed with the ridge tracking algorithm.}
 \label{fig:abc}
\end{figure}
The LCS are shown in Fig. \ref{fig:abc} and agree with previously published results~\cite{HallerG:01a,PadbergK:09a}. Fig. \ref{fig:abc} also shows single slice of the FTLE field well as the FTLE field and corresponding LCS at a height of $z=\pi$. The complex LCS present in the ABC flow are a product of the non-trivial invariant manifolds of this flow. They clearly divide the flow into different regions that appear as tube-like structures through the flow domain. These tubes are dynamically distinct from one another and particles travel within and along the tubes without escaping to other regions of the space.

\subsection{Swimming jellyfish}

As a final example, we compute the LCS created by a jetting type jellyfish. The jellyfish body motion was extracted from a digital video of a swimming individual jellyfish (species {\it Sarsia tubulosa}) and use as input to an arbitrary Lagrangian-Eulerian CFD code that computes the resulting flow field and jellyfish acceleration. Full details of this procedure can be found in Sahin and Mohseni~\cite{Mohseni:09b} and Sahin et al.~\cite{Mohseni:09c}.

The resulting axisymmetric velocity field is returned on a moving, non-uniform quadrilateral mesh in $(r,z)$-coordinates. During LCS computations, the $(x,y,z)$ coordinates are converted to $(r,z)$ coordinates to compute the velocity and the the velocity is converted back to Cartesian coordinates for particle advections. The mesh type creates significant complications for velocity interpolation during particle advection since even locating the mesh element that contains a given point is non-trivial. To address this issue efficiently, an alternating digital tree (ADT) is used to search the domain for the element that contains a given drifter particle~\cite{Bonet:91}. The ADT recursively divides the space in half so that at each node in the tree, only one branch must be searched. Since the nodes of the mesh elements are listed in counter clockwise order and the elements are convex, a point $\mathbf{p}$ is inside (including the boundary) an element if and only if
\[
\hat{\mathbf{z}} \cdot [ (\mathbf{v}_i-\mathbf{p}) \times (\mathbf{v}_{(i~{\rm mod}~4)+1}-\mathbf{v}_i) ]\ge 0~\forall~i \in \{1,2,3,4\}
\]
for vertices $\{\mathbf{v}_i\}$.

Once the element containing a particle is found, the velocity must be interpolated onto the drifter particle. Since the elements are generally none rectangular, simple linear interpolation is not possible. Instead, perspective projection is used to map the quadrilateral element and the point of interest onto the unit square. Bilinear interpolation is then used to approximate the velocity of the particle.

\begin{figure}
\centering

\begin{minipage}[b]{0.1\linewidth}
\centering
\topskip0pt
\vspace*{\fill}
\bf (a)
\vspace*{\fill}
\end{minipage}
\begin{minipage}[b]{0.8\linewidth}
\centering
 \includegraphics[width=1\linewidth]{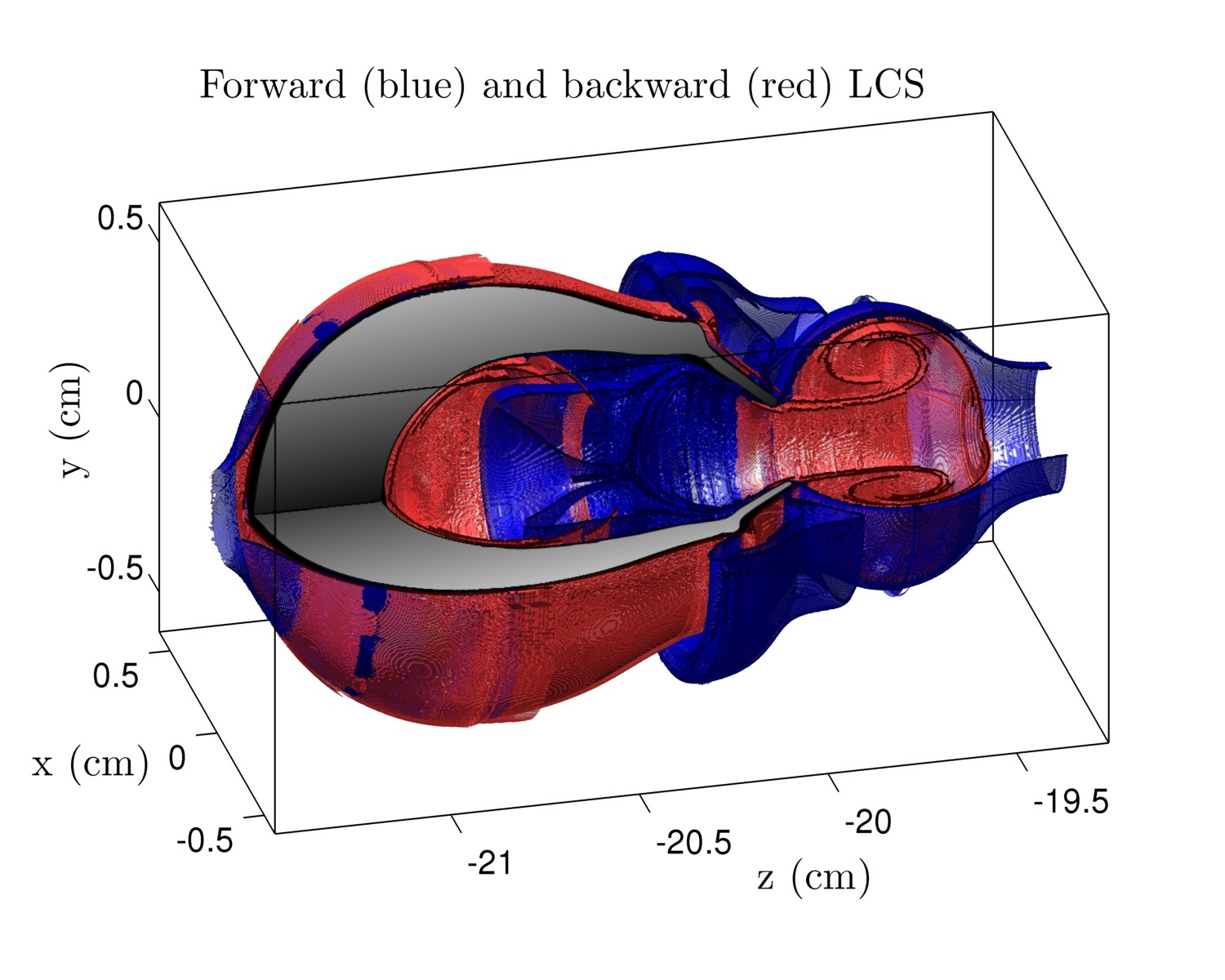}
\end{minipage}\\

\begin{minipage}[b]{0.1\linewidth}
\centering
\bf (b)
\end{minipage}
\begin{minipage}[b]{0.8\linewidth}
\centering
 \includegraphics[width=1\linewidth]{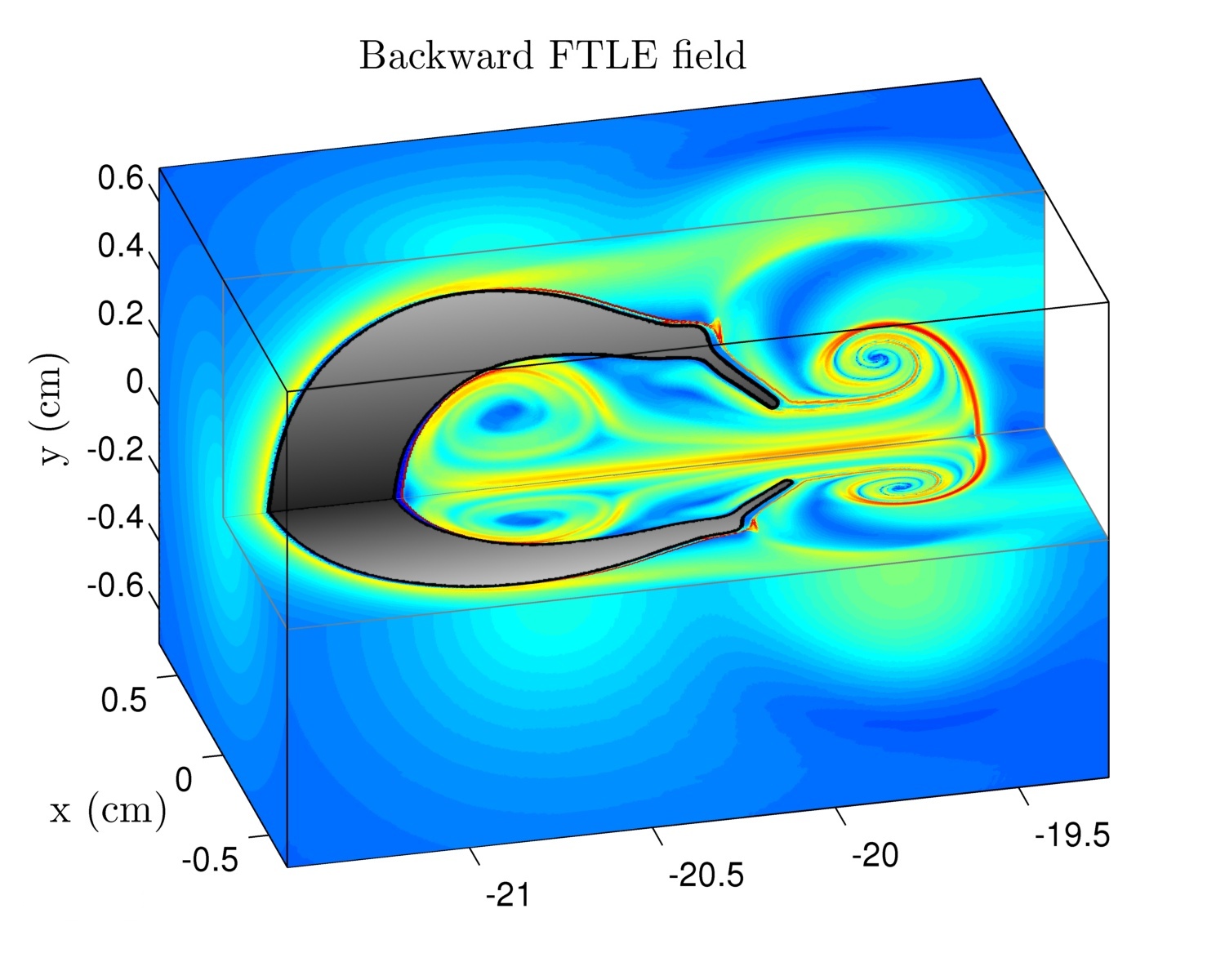}
\end{minipage}

 \caption{(a) The results of the ridge tracking algorithm and (b) the backward FTLE field for the swimming jellyfish. A strong vortex is being ejected near the end of the jellyfish's bell contraction. A cutaway view is shown so that the full LCS structure is visible.}
 \label{fig:sarsia}
\end{figure}
The jellyfish chosen for this investigation has a swimming period of 1 s between contractions and an integration time of 0.5 s was used to compute the LCS. Fig. \ref{fig:sarsia} shows the LCS as computed with the ridge tracking algorithm as well as the backward FTLE field computed with the standard FTLE algorithm. The backward LCS (Fig. \ref{fig:sarsia}a) clearly show a strong vortex being ejected as the jellyfish's bell contraction comes to an end. Additionally, the forward LCS outline fluid ahead of the vortex that will soon be entrained into the vortex ring as well as a region of fluid near the jellyfish bell that will be drawn into the bell during the relaxation phase of jellyfish swimming. These results are in excellent agreement with previously published LCS for this jellyfish~\cite{Mohseni:09d}. Furthermore, this example clearly demonstrates the visualization advantages offered by the ridge tracking algorithm. Since the actual LCS surface are computed it is much simpler to visualize both the forward and backward LCS simultaneously.

\subsection{Timing results}
\label{sec:timing_results}

We expect the computational time of the ridge tracking algorithm to be $\mathcal{O}(1/dx^2)$ for a grid of spacing $dx$ since computations are performed only near the 2D ridge surfaces. To establish the computational order, we have computed the full FTLE field using a standard FTLE algorithm and also computed the LCS with the ridge tracking algorithm presented above for the double gyre, ABC, and jellyfish flows. Computations were performed on at a variety of grid resolutions as well as on a single core and with a parallel code running on 16 or 48 cores. A least squares best fit was performed on a log-log scale for each case, assuming a fit of
\begin{equation}
t_f=C/(dx^{\alpha}).
\end{equation}
The resulting data points and curve fits are shown in Fig. \ref{fig:data} and show that the standard algorithm is $\approx\mathcal{O}(1/dx^{3.0})$ while the ridge tracking algorithm scales approximately as $\mathcal{O}(1/dx^{2.1})$. This performance is maintained for the parallel version of the code.

It is worth noting that in the jellyfish example, for low resolutions the time required to read the velocity data files and build the ADTs used for search at each time step is a significant part of the total computational time. Since the velocity read in time does not change with grid resolution, this has the effect of giving artificially low exponents for the algorithm order (both for the standard algorithm and the ridge tracking algorithm). To compensate for this effect and more accurately estimate the asymptotic order, we estimated the velocity read in and ADT creation time and subtracted this value from the total run time before computing the computational order for the jellyfish example. These modified times are reported in Fig. \ref{fig:data}d. The resulting values for $\alpha$ closely match the values that would result from using only the last few data points and represent a closer approximation of the asymptotic values of $\alpha$.

\begin{figure}
\centering
\begin{minipage}[b]{1\linewidth}
\centering
 \includegraphics[width=.48 \textwidth]{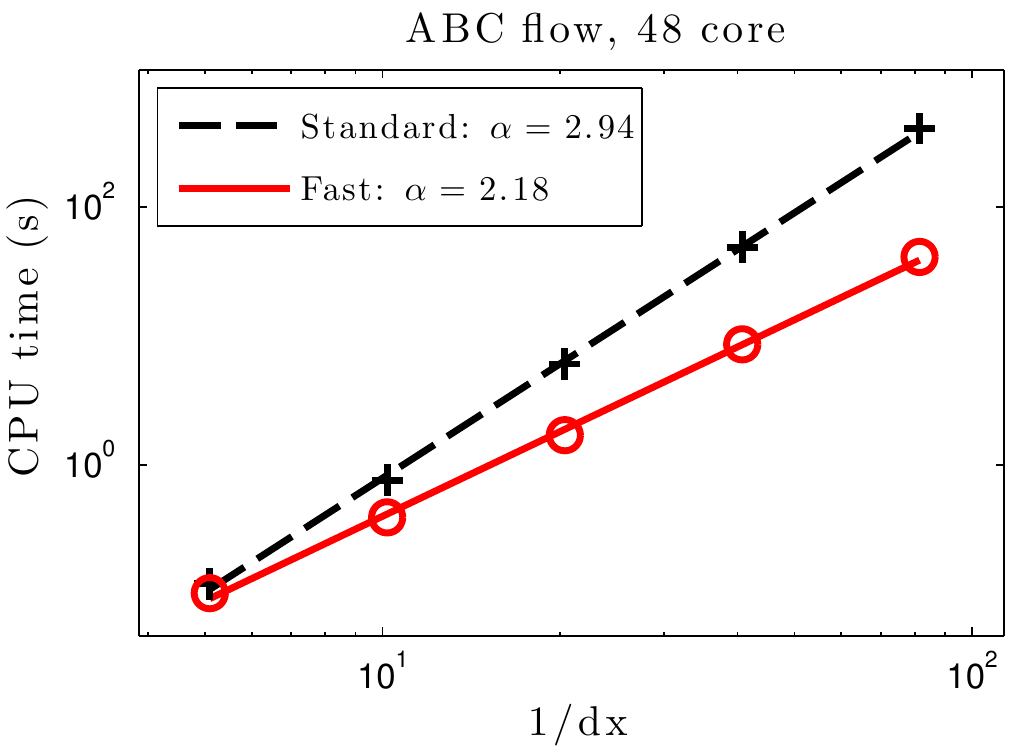}
 \includegraphics[width=.48 \textwidth]{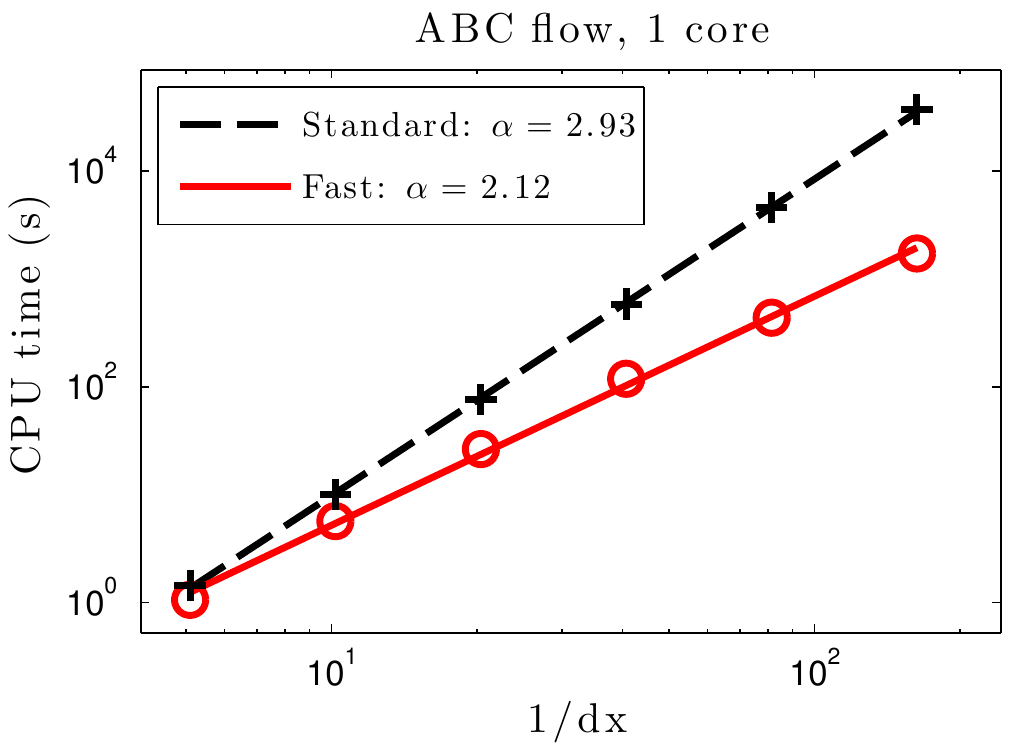}
 \end{minipage}
 
 \begin{minipage}[b]{.48\linewidth}
 \centering
 \footnotesize
 \bf (a)
 \end{minipage}
 \begin{minipage}[b]{.48\linewidth}
 \centering
 \footnotesize
 \bf (b)
 \end{minipage}
 
 \begin{minipage}[b]{1\linewidth}
\hfill
 \end{minipage}
 
\begin{minipage}[b]{1\linewidth}
\centering
 \includegraphics[width=.48 \textwidth]{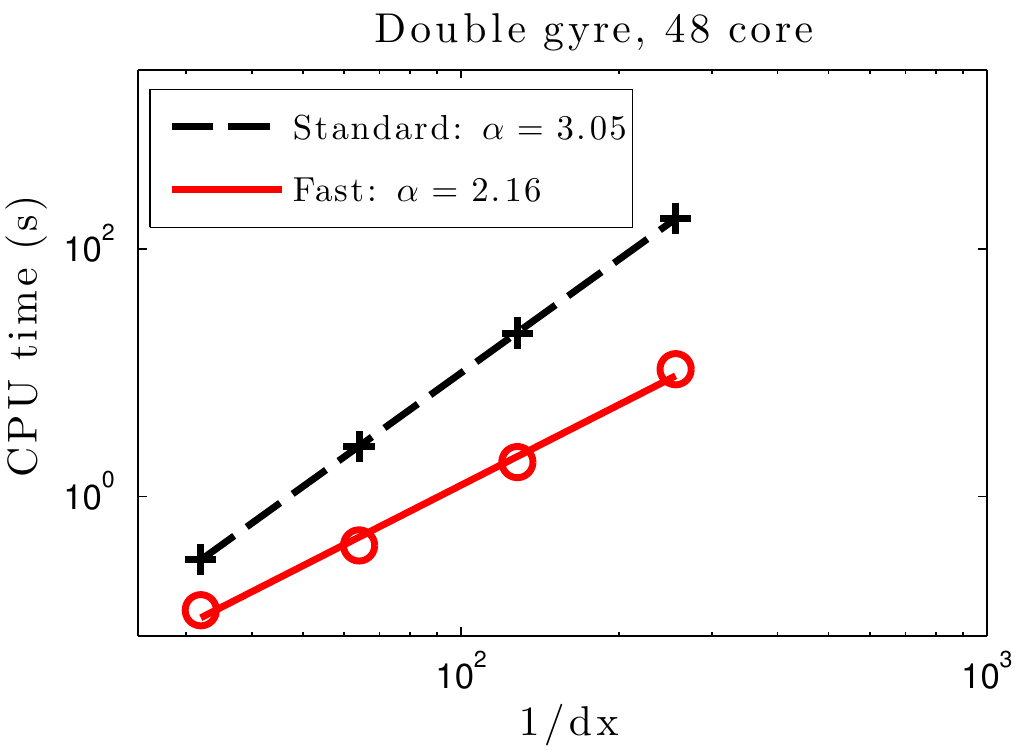}
 \includegraphics[width=.48 \textwidth]{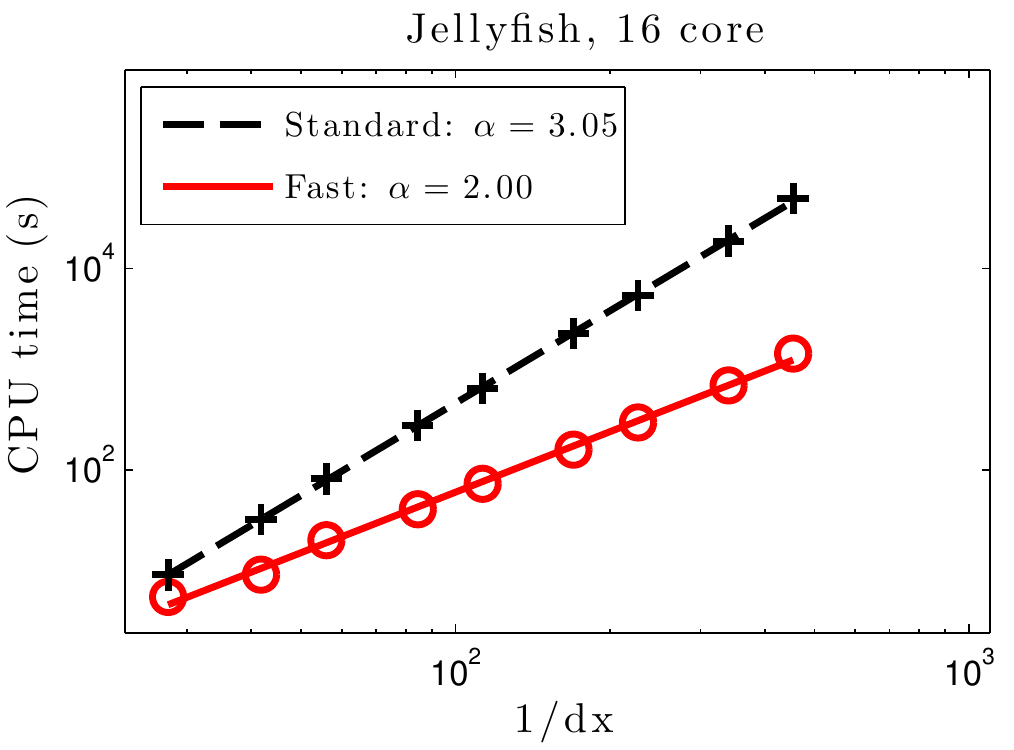}
 \end{minipage}
 
 \begin{minipage}[b]{.48\linewidth}
 \centering
 \footnotesize
 \bf (c)
 \end{minipage}
 \begin{minipage}[b]{.48\linewidth}
 \centering
 \footnotesize
 \bf (d)
 \end{minipage}
 \caption{Timing results for the standard FTLE algorithm and the ridge tracking algorithm. $\alpha$ is the scaling exponent of the algorithm (CPU time $\approx \mathcal{O}(1/dx^\alpha)$) and appears as the slope of the lines in these $\log$-$\log$ plots. The standard algorithm scales as $\mathcal{O}(1/dx^{3.0})$ and the ridge tracking algorithm scales as $\mathcal{O}(1/dx^{2.1})$.}
 \label{fig:data}
\end{figure}

It is also expected that the computational time should be directly proportional to the surface area of the LCS in the domain. We test this by generating an artificial FTLE field that has ridges along pre-selected planes. The planes are defined by $z=0.1(x+y)+h_i$ and the artificial FTLE field is given by
\begin{equation}
\sigma(x,y,z) = \sum_i \exp[-2000(z-0.1(x+y)-h_i)^2].
\label{eq:fake_ftle}
\end{equation}
Each plane that defines an FTLE ridge is tilted slightly so that it is out of line with the grid to make detecting and tracking the ridge slightly more realistic. We use the domain $[0,1]^3$ where each ridge has an area of $A_{LCS}=\sqrt{1.02}$ for $h_i\in (0,0.8)$ and test six different cases corresponding to $1-6$ ridges in the domain. Particles are advected using the double gyre velocity field listed above to account for the particle advection time, but instead of returning the true FTLE field for the double gyre flow, the artificial field of Eq. \ref{eq:fake_ftle} with pre-selected ridges is returned. The results are listed in Table \ref{tab:SA} which shows the computational time according to surface area for four different values of $dx$.

\begin{table}
\centering
\begin{tabular}{cc||c|c|c|c|c|c|}
\hline
\multicolumn{2}{|c||}{ \multirow{2}{*}{CPU time (s)}} & \multicolumn{6}{|c|}{\bf Surface area} \\
\cline{3-8}
 \multicolumn{2}{|c||}{\multirow{1}{*}{}}  & \cellcolor[gray]{0.95} {\bf 1.01} & \cellcolor[gray]{0.95} {\bf 2.02} & \cellcolor[gray]{0.95} {\bf 3.03} & \cellcolor[gray]{0.95} {\bf 4.04} & \cellcolor[gray]{0.95} {\bf 5.05} & \cellcolor[gray]{0.95} {\bf 6.06} \\
\hline
\hline
\multicolumn{1}{|c|}{\multirow{4}{*}{\bf Grid Spacing}} & \cellcolor[gray]{0.95} {\bf 1/64} & 0.533 & 0.689 & 0.844 & 0.987 & 1.162 & 1.333 \\
\cline{2-8}
 \multicolumn{1}{|c|}{} & \cellcolor[gray]{0.95} {\bf 1/128} & 2.713 & 3.306 & 3.913 & 4.625 & 5.385 & 6.121 \\
\cline{2-8}
 \multicolumn{1}{|c|}{} & \cellcolor[gray]{0.95} {\bf 1/256} & 19.54 & 22.12 & 25.10 & 27.68 & 32.79 & 34.80 \\
\cline{2-8}
 \multicolumn{1}{|c|}{} & \cellcolor[gray]{0.95} {\bf 1/512} & 192.8 & 204.5 & 215.7 & 235.7 & 244.7 & 259.9 \\
\hline
\end{tabular}
\caption{CPU time (in seconds) to compute the artificial FTLE ridges of Eq. \ref{eq:fake_ftle} for various surface areas and grid spacings.}
\label{tab:SA}
\end{table}

All four values of $dx$ show a linear relationship between LCS surface area and CPU time and least squares fits result in the following regression coefficients for the fit $t_{CPU} = C_1 + C_2 A_{LCS}$ where $A_{LCS}$ is the surface area of the LCS and $t_{CPU}$ is the required CPU time:
\begin{center}
\begin{tabular}{|c|c|c|}
\hline
$dx$ & $C_1$ & $C_2$\\
\hline
1/64 & 0.368 & 0.157 \\
1/128 & 1.94 & 0.679 \\
1/256 & 15.9 & 3.14 \\
1/512 & 177.9 & 13.5 \\
\hline
\end{tabular}
\end{center}

The relatively large values of $C_1$ for all these cases means that there is some initial cost regardless of the amount of LCS surface area. This is due to the cost associated with initializing data structures and performing the initial ridge detection step as described in Section \ref{sec:initial_detection}. Additionally, $C_1$ appears to scale roughly as $\mathcal{O}(1/dx^3)$. Since the current implementation of the ridge tracking code uses and allocates full 3D arrays rather than using sparse data structures, it is reasonable to expect this relationship. On the other hand, $C_2$ scales roughly as $\mathcal{O}(1/dx^2)$. This accounts for the majority ridge tracking part of the algorithm.

The $\mathcal{O}(1/dx^{3.0})$ scaling of $C_1$ may explain why the overall computational order of the algorithm is $\mathcal{O}(1/dx^{2.1})$ rather than $\mathcal{O}(1/dx^2)$. As resolution is increased beyond current capabilities, the initialization cost of may be expected to completely overwhelm the ridge tracking cost due to this difference in order. The implementation of sparse data structures would likely help solve this problem.

\section{Conclusions}

As the problems being analyzed with LCS techniques become increasingly complex, the corresponding computations become increasing expensive. The ridge tracking algorithm presented in this paper has been shown to reduce the order of LCS computations from $\mathcal{O}(1/dx^{3.0})$ to about $\mathcal{O}(1/dx^{2.1})$ for three-dimensional flows. This reduction in order allows potentially tremendous savings in computational time as the required LCS resolution is increased.

The effectiveness and algorithm properties have been demonstrated by several examples, including the analytically defined double gyre and ABC flow and the swimming jellyfish that is defined by velocity stored in data files. On single processor as well as multicore machines, the ridge tracking algorithm shows the expected change in computational order and provides large speed ups for all tested cases.

We have also proved that although the ridge tracking algorithm detects {\it coordinate local maxima} rather than the actual ridges, for well defined ridges the associated error is small. The distance between a ridge and the surfaces detected by this algorithm is $\mathcal{O}(\norm{\mathbf{D}F}/\abs{\lambda_1})$ where $\mathbf{D}F$ is the gradient along the ridge and $\abs{\lambda_1}$ is the second derivative normal to the ridge (or, equivalently the smallest eigenvalue of the Hessian of $F$). In typical examples this error is at smaller than the grid spacing and for well defined ridges it may be several orders of magnitude smaller than the grid spacing.

The general framework of this algorithm could easily be adapted to other LCS definitions or techniques. For example, the finite size Lyapunov exponent (FSLE) could easily be used in lieu of the FTLE. A 2D version of this algorithm (or the grid-less algorithm of Lipinski and Mohseni~\cite{Mohseni:10g}) could be relatively easily modified to use the variational formulation for LCS which was recently proposed by Haller~\cite{HallerG:11a,HallerG:11b}.

Future work will focus on implementing sparse data structures to further reduce the computational cost of large LCS calculations and on implementing some addition LCS techniques such as those mentioned in the previous paragraph. We expect that the use of sparse data structures may further reduce the order of computation from $\mathcal{O}(1/dx^{2.1})$ nearer to $\mathcal{O}(1/dx^{2.0})$.

\bibliographystyle{model1-num-names}
\bibliography{RefA1}

\appendix

\section{Additional proofs}

This appendix contains additional theorems and proofs referenced by the proofs of the theorems in the text.

\begin{theorem}
\label{thm:closest_vector}
Given an orthonormal basis $\{\mathbf{e}_i\}$ for $\mathbb{R}^n$ and an arbitrary unit vector $\mathbf{v}$ there is a basis vector $\mathbf{e}_j$ that maximizes $\abs{ \mathbf{v} \cdot \mathbf{e}_i }$ and $\abs{ \mathbf{v} \cdot \mathbf{e}_j } \ge 1/\sqrt{n}$.
\end{theorem}
\begin{proof}
Existence of the maximum is implied by the extreme value theorem since the dot product is a continuous function and the set $\{\mathbf{e}_i\}$ is finite and therefore closed and bounded. The other part of the proof is as follows:
Assume $\abs{ \mathbf{v} \cdot \mathbf{e}_j } <1/\sqrt{n}$. Then
\begin{align*}
1=\norm{ \mathbf{v} } & = \norm{ \sum_{i=1}^n(\mathbf{e}_i \cdot \mathbf{v})\mathbf{e}_i } \\
& = \sqrt{(\sum_{i=1}^n(\mathbf{e}_i \cdot \mathbf{v})\mathbf{e}_i)\cdot(\sum_{i=1}^n(\mathbf{e}_i \cdot \mathbf{v})\mathbf{e}_i)}\\
& = \sqrt{\sum_{i=1}^n(\mathbf{e}_i \cdot \mathbf{v})^2}\\
& \le \sqrt{\sum_{i=1}^n(\mathbf{e}_j \cdot \mathbf{v})^2} = \sqrt{n(\mathbf{e}_j \cdot \mathbf{v})^2}\\
& < \sqrt{n(1/\sqrt{n})^2} \\
& < 1
\end{align*}
which is a contradiction so $\abs{ \mathbf{v} \cdot \mathbf{e}_j } \ge 1/\sqrt{n}$.
\end{proof}

\begin{theorem}
 \label{thm:quadratic}
 Given a quadratic equation $y=ax^2+bx+c$ with bounds on the real valued constants 
\begin{align*}
 & \abs{c} \le c_1 \\
 0 < b_0 \le & \abs{b} \le b_1 \\
 & \abs{a} \le \dfrac{b_0^2}{4c_1}
\end{align*}
there is a real root $x^*$ such that $\abs{x^*}<\dfrac{2c_1}{b_0}$.
\end{theorem}
\begin{proof}
 If $a=0$, there is a single root at $x^*=-c/b$ and $\abs{x^*}=\abs{-c/b}<2c_1/b_0$.\\
 If $a\ne0$, the roots are given by the alternate form of the quadratic formula
$$
 x = \dfrac{2c}{-b\pm\sqrt{b^2-4ac}}.
$$
The discriminant, $\Delta=b^2-4ac>b_0^2-4\dfrac{b_0^2}{4c_1}c_1=0$ is positive so the quadratic has two real roots. Denote these roots
\begin{align*}
 x_1 =  \dfrac{2c}{-b + \sqrt{b^2-4ac}} \\
 x_2 =  \dfrac{2c}{-b - \sqrt{b^2-4ac}}.
\end{align*}
If $b\le0$ let $x^*=x_1$, otherwise let $x^*=x_2$. Then
\begin{align*}
\abs{x^*} & = \dfrac{2\abs{c}}{\abs{b}-\sqrt{b^2-4ac}} \\
 & < \dfrac{2c_1}{b_0}
\end{align*}

\end{proof}

\end{document}